\documentclass[draftclsnofoot,onecolumn,12pt]{IEEEtran}

\usepackage{bbm}
\usepackage{amsmath,graphics,amssymb,epsfig,subfigure,color,amsthm,cite}
\usepackage{array}
\usepackage{multirow}
\usepackage{enumerate}
\usepackage{algorithmic}
\usepackage{algorithm}
\usepackage{bm}
\usepackage[cmintegrals]{newtxmath}
\usepackage{tablefootnote}

\newtheorem{lem}{Lemma}

\begin{document}
\title{Joint User Selection, Power Allocation, and Precoding Design with Imperfect CSIT for Multi-Cell MU-MIMO Downlink Systems}
%\title{Limits of Linear Processing for Large-Scale MU-MIMO Systems Using Principal Component Analysis }

\author{\IEEEauthorblockN{Jiwook Choi, Namyoon Lee, Song-Nam Hong, and Giuseppe Caire}\\
	\thanks{J. Choi and N. Lee are with the Department of Electrical Engineering, POSTECH, Pohang, Gyeongbuk 37673, South Korea  (e-mail: \{jiwook, nylee\}@postech.ac.kr).}
	\thanks{S.-N. Hong is with the Department of Electrical and Computer Engineering, Ajou University, Suwon, Gyeonggi 16499, South Korea (e-mail: snhong@ajou.ac.kr).}
		\thanks{G. Caire is with the Department of Electrical and Computer Engineering, Technical University of Berlin, Berlin, Germany, (e-mail: caire@tu-berlin.de).}
			\thanks{A part of this paper was presented in \cite{Choi_Lee_Hong_Caire}.}
%\IEEEauthorblockA{
%Department of Electrical Engineering, POSTECH, \\Pohang. Gyeongbuk, Korea\\
%Emails: \{nylee,jiwook\}@postech.ac.kr} \\
%\IEEEauthorblockA{Ajou University, Suwon, Korea,\\
%              email:  snhong@ajou.ac.kr}\\
%\IEEEauthorblockA{ Technical University of Berlin, Berlin, Germany,\\
%              email: caire@tu-berlin.de}\\
}

%\author{\IEEEauthorblockN{Jiwook Choi, Namyoon Lee, Song-Nam Hong, and Giuseppe Caire}\\
%\IEEEauthorblockA{
%Department of Electrical Engineering, POSTECH, \\Pohang. Gyeongbuk, Korea\\
%Emails: \{jiwook,nylee\}@postech.ac.kr} \\
%\IEEEauthorblockA{Ajou University, Suwon, Korea,\\
%              email: snhong@ajou.ac.kr}\\
%\IEEEauthorblockA{ Technical University of Berlin, Berlin, Germany,\\
%              email: caire@tu-berlin.de}\\
%}
%Wireless Communications and Sensing Research Laboratory\\
\maketitle
\begin{abstract}
In this paper, a new optimization framework is presented for the joint design of user selection, power allocation, and precoding in multi-cell multi-user multiple-input multiple-output (MU-MIMO) systems when imperfect channel state information at transmitter (CSIT) is available. By representing the joint optimization variables in a higher-dimensional space, the weighted sum-spectral efficiency maximization is formulated as the maximization of the product of Rayleigh quotients. Although this is still a non-convex problem, a computationally efficient algorithm, referred to as generalized power iteration precoding (GPIP), is proposed. The algorithm converges to a stationary point (local maximum) of the objective function and therefore it guarantees the first-order optimality of the solution. By adjusting the weights in the weighted sum-spectral efficiency, the GPIP yields a joint solution for user selection, power allocation, and downlink precoding. The GPIP is also extended to a multi-cell scenario, where cooperative base stations perform joint user selection and design their precoding vectors by sharing global yet imperfect CSIT within the cooperative BSs. System-level simulations show the gains of the proposed approach with respect to conventional user selection and linear downlink precoding.

\end{abstract}
\IEEEpeerreviewmaketitle
\section{Introduction}

\subsection{Motivation}
In cellular networks, a major bottleneck in achieving a high spectral efficiency is interference. Compared to conventional single-antenna downlink cellular networks, a multi-user multiple-input multiple-output (MU-MIMO) downlink cellular network suffers from inter-user-interference (IUI) in addition to the usual inter-cell-interference (ICI) \cite{Marzetta_2010,Lozano_2013,Lee_Lozano_2014}. To alleviate the interference in MU-MIMO systems, accurate knowledge of channel state information (CSI) at base stations (BSs) is indispensable. In practice, however, obtaining perfect CSI at the BSs is infeasible in MU-MIMO systems. In frequency division duplex (FDD) systems, imperfect CSI can be obtained via downlink training and feedback, and its accuracy is determined by the feedback rate and by the delay introduced by the probing and feedback loop \cite{Love_2008,Jindal_2006,Yoo_2006}. In time division duplex (TDD) systems, the BSs acquire CSI from the uplink pilot signals with appropriate transceiver hardware calibration \cite{Hassibi_Hochwald_2003,Caire_Ravindran_2010,Jose_Vishwanath_2011,Yin_Gesbert_2013}. In particular, when non-orthogonal uplink pilots are used across multiple cells, the uplink channel measurements are contaminated by co-pilot interference across different cells. This effect is particularly significant in massive MIMO systems, in which the number of antennas at the BS, $N$, is much larger than the number of downlink data streams, $K$, i.e., $N\gg K$ \cite{Jose_Vishwanath_2011,Yin_Gesbert_2013}.

In MU-MIMO downlink cellular networks, finding a jointly optimal user selection, precoding, and power allocation solution that maximizes the weighted sum-spectral efficiency is a very challenging problem even with the case of perfect CSI at transmitter (CSIT). The major hindrance in the design is that the downlink signal-to-interference-plus-noise-ratio (SINR) of a user depends on the set of scheduled users and precoding/power allocation of the other users; thereby, the SINRs of the users are interwoven with each other. The joint optimization of user-selection, power allocation, and precoding vectors in order to maximize the weighted sum-spectral efficiency is known to be NP-hard\cite{Luo_Zhang_2008}. When considering imperfect CSIT, the problem becomes even more complicated since an exact expression of the achievable spectral efficiency in the presence of non-perfect CSIT is not available, and one must resort to some bound on the corresponding ergodic spectral efficiency \cite{Caire_2018}. In addition, the effect of the CSIT errors should be appropriately taken into account in the optimization problem, so that the solution is robust to imperfect CSIT. In this paper, we present a novel optimization framework that provides a computationally efficient heuristic solution of the weighted sum-spectral efficiency maximization problem under  imperfect CSIT.

\subsection{Prior Works}
In the context of single-cell MU-MIMO systems, the user selection problem for a given precoding strategy has been extensively studied in the past decade \cite{Yoo_Goldsmith_2005,Dimic_2005}. Semi-orthogonal user selection with zero-forcing precoding (SUS-ZF) \cite{Yoo_Goldsmith_2005} is perhaps the most representative and widely used. The key idea of SUS-ZF is to find a set of users whose channel directions are nearly orthogonal (semi-orthogonal) to achieve a high sum-spectral efficiency with ZF precoding. This method was shown to achieve the same scaling law of the sum-spectral efficiency obtained by ZF-dirty-paper-coding (ZF-DPC) \cite{Caire_Shamai_2003} when $K$ is sufficiently larger than $N$, i.e., $N/K\ll 1$ with computational complexity $\mathcal{O}\left(KN^2\right)$. Prior studies  \cite{Yoo_Goldsmith_2005,Dimic_2005}, however, focused on the user selection problem for a fixed precoding strategy instead of the joint design of them conjunction with power allocation. The use of uplink-downlink duality is another common approach to find a joint solution of precoding vectors and power allocation in single-cell MU-MIMO systems with perfect CSIT \cite{Viswanath_Tse_2003,Viswanath_Goldsmith_2003}. Using this approach, the optimal precoding and power allocation solution to minimize the total transmit power subject to SINR constraints \cite{Schubert_Boche_2004,Yu_Lan_2007} was proposed. The use of this duality-based precoding and power allocation was applied together with a heuristic user selection algorithm \cite{Song_Cruz_2008}. While the total transmit power minimization subject to individual SINR constraints is a convex problem, it is well-known that the sum-spectral efficiency maximization subject to a total power constraint under linear precoding is non-convex, even in the case of perfect CSIT. In \cite{Stojnic_Hassibi_2006}, a heuristic gradient update for the maximization of the sum-spectral efficiency that obtains directly the precoding vectors and the corresponding user power allocation was proposed and shown to converge to a local maximum. This approach, however, is not applicable to imperfect CSIT and/or a multi-cell environment. Considering the multi-cell MU-MIMO setting, we should distinguish between different levels of cooperation.  When all the antennas of the BSs are jointly precoded in a centralized fashion, the system reduces to a single giant cell with distributed antennas (the so-called cloud radio access network (C-RAN) architecture \cite{Aguerri_2017}).  In this paper, we consider an intermediate form of cooperation where the BSs share their CSIT and jointly optimize the beam forming vectors, while each BS serves uniquely its own users. In this case, BS cooperation restricts to coping with inter-cell interference. For this scenario, in \cite{Gesbert_Yu_2010,Huh_Caire_2011,Huh_Caire_2012,Huh_Caire_Ramprashad_2012,Yu_Kwon_2013,Huang_Rao_2013}, joint user selection, power allocation, and precoding algorithms for BS cooperation have been proposed to effectively mitigate inter-cell-interference. These algorithms can be applicable to the case of multi-cell MU-MIMO systems using single-cell operation. The major limitation of these studies, however, is that they assumed perfect CSIT; yet, imperfect CSIT assumption is more practically relevant.  

  Under the imperfect CSIT assumption, robust MU-MIMO transmission strategies have been extensively studied in \cite{Joudeh_Clerckx_2016_TC, Joudeh_Clerckx_2016_TSP, Dai_Clerckx_2016, MAT_2012, Lee_Heath_2014,Lee_Heath_2015}. In \cite{MAT_2012, Lee_Heath_2014,Lee_Heath_2015}, a set of new transmission strategies proposed when CSIT is completely-delayed \cite{MAT_2012} or moderately-delayed \cite{ Lee_Heath_2014,Lee_Heath_2015}. The underlying limitation of these studies is that they optimize the degrees-of-freedom (DoF) for the MU-MIMO systems, which is of limited use in practical finite SNR conditions. In addition, using a rate-splitting approach, linear precoding methods were developed under the imperfect CSIT assumption \cite{Joudeh_Clerckx_2016_TC, Joudeh_Clerckx_2016_TSP, Dai_Clerckx_2016}. All the aforementioned studies, however, only focused on the precoding design; the user selection and power allocation methods are not jointly taken into account together with multi-user precoding. 

\subsection{Contributions}
 We consider a multi-cell downlink MU-MIMO system in which a BS equipped with $N$ antennas serves $K$ downlink users, each with a single-antenna. The main contributions of this paper are summarized as follows:
 
\begin{itemize}
	\item  We propose a novel optimization framework to jointly solve the user selection, power allocation, and precoding design problem for multi-cell MU-MIMO downlink systems. Specifically, using the concept of generalized mutual information (GMI) introduced in \cite{Shamai_2002,    Ding_Blostein_2010, Medard_2000,Yoo_Goldsmith_2006}, we derive a lower bound of the weighted sum-spectral efficiency when each BS has knowledge of imperfect and local CSIT for its own downlink users.  Then, the proposed optimization framework is to reformulate the weighted sum-spectral efficiency maximization (WSM) problem (which is well-known as an integer-mixed optimization problem in \cite{Luo_Zhang_2008}) into the maximization problem of the product of the Rayleigh quotients. Although this reformulated optimization problem is still a non-convex optimization problem, it is more tractable by representing the joint optimization variables in a higher-dimensional space, which allows to effectively remove the integer constraints in the user-selection sub-problem. 
	\item We also present an algorithm that quickly converges to a solution satisfying the first-order optimality condition of the reformulated optimization problem. To accomplish this, we derive the first-order Karush-Kuhn-Tucker (KKT) condition of our reformulated problem. The derived KKT condition constitutes a set of $NK$ non-linear equations with $NK$ unknown variables; finding the solution that satisfies the KKT condition, in general, needs a very high computational complexity. By interpreting the problem that finds the solution as a class of functional generalized eigenvalue problems, we propose a computationally efficient algorithm, which is referred to as the \textit{generalized power iteration precoding  (GPIP)}. The key idea of the GPIP is to find the principal component of the functional generalized eigenvalue problem in an iterative fashion. One important remark is that the proposed GPIP method quickly converges to the solution with the first-order optimality for the joint design problem of user selection, power allocation, and precoding. 
	%In addition, we show that the computational complexity of the power iteration precoding is the order of $\mathcal{O}\left(JKN^2\right)$ by harnessing some special structures in the optimization problem, where $J$ is the number of iterations. This complexity order is the same as that of ZF precoding $\mathcal{O}\left(K^3\right)$ when $K=N$, which does not perform user selection and power control in contrast to the proposed method.
	
	\item In addition, we extend our optimization framework to a multi-cell cooperation scenario, in which a set of cooperative BSs takes inter-cell interference into account as colored noise and shares imperfect CSIT within the cooperative BS cluster. Similar to the non-cooperative transmission case, we reformulate the WSM problem into the maximization of the product of Rayleigh quotients for the multi-cell cooperative transmission. This fact shows that our optimization framework is applicable to a more general scenario regardless of the number of antennas per BS, users per cell, and cooperative BSs. By modifying the proposed GPIP algorithm for the non-cooperative transmission, we present a multi-cell precoding method, which jointly finds a set of cooperatively scheduled users, power allocation, and precoding solutions to effectively control inter-cluster-interference under the individual BS power constraint.
	
	\item By simulations, we demonstrate that, in the case of the single-cell MU-MIMO system with perfect CSIT, the proposed algorithm achieves the same sum-spectral efficiency with that of ZF-DPC in a low SNR regime regardless of the number of BS antennas, $N$, and the number of users, $K$. This result implies that linear precoding can be sufficient to achieve a near sum-capacity under a certain condition, provided that user selection and power allocation are jointly performed with linear precoding. Considering the case of imperfect CSIT, we also demonstrate that the proposed algorithm offers a considerable sum-spectral efficiency gain over the existing scheduling and precoding methods by simulations. To gauge the gains of the proposed solution in practical systems, we evaluate the ergodic sum-spectral efficiency through system-level-simulations. It is observed that the proposed algorithm provides a noticeable spectral efficiency gain over the conventional user-selection and precoding methods under imperfect CSIT from a system-level perspective. We show that the sum-spectral efficiency improves when the number of users per cell increases for a fixed BS antennas and when the number of BS antennas increases for a fixed number of users per cell. Lastly, we demonstrate that the proposed multi-cell cooperative transmission method provides a significant spectral efficiency gain over the non-cooperative transmission methods under imperfect CSIT. 
	
\end{itemize}

% 
% \subsection{Organization and Notation}
% 
% 
%The rest of the paper is organized as follows. The system model and CSIT assumptions are described in Section II. In Section III, the proposed optimization framework is formulated. The weighted sum-spectral efficiency maximization algorithm is proposed, and important remarks of the proposed algorithm are presented in Section IV. Simulation results are presented in Section V, and Section VI concludes the paper with possible future work.
%
%Boldface uppercase letters denote matrices, boldface low-ercase letters denote column vectors and standard letters denote scalars. The superscrips $(\cdot)^{\sf T}$ and $(\cdot)^{\sf H}$ denote transpose and conjugate-transpose (Hermitian) operators, respectively. $\mathbb{E}_X[\cdot]$ denotes the expectation w.r.t the random variable $X$.

 %%%%%%%%%%%%%%%%%%%%%%%%%%%%%%%%%%%%%%
\section{System Model}
This section presents a multi-cell MU-MIMO system model and the corresponding ergodic sum-spectral efficiency.
 
\subsection{Network and Channel Model}
We consider a MU-MIMO cellular network consisting of $L$ cells. Each cell consists of one BS equipped with $N$ antennas and $K$ users equipped with a single antenna. We denote the downlink channel vector from the $j$th BS to the $k$th user in the $\ell$th cell by ${\bf h}_{j,\ell,k}\in\mathbb{C}^{N\times 1}$, and it is assumed to be a Rayleigh fading process, i.e., ${\bf h}_{j,\ell,k} \sim \mathcal{CN}\left({\bf 0},{\bf R}_{j,\ell,k}\right)$ where ${\bf R}_{j,\ell,k}=\mathbb{E}\left[{\bf h}_{j,\ell,k}{\bf h}_{j,\ell,k}^{\sf H}\right]\in\mathbb{C}^{N\times N}$ is the channel correlation matrix. This channel correlation matrix captures the macroscopic effects of the channel.  We consider a spatially correlated channel model to reflect the spatial correlation effect among BS antennas.  In particular, a geometric one-ring scattering model is considered as in \cite{Clerckx2013_book}. We denote the azimuth angle and the angular spread of the $k$th user in the $\ell$th cell with respect to the orientation perpendicular to the array axis of the $j$th antennas by $\theta_{j,\ell,k}$ and $\Delta_{j,\ell,k}$, respectively. In addition, the average large-scale fading from the $j$th BS to the $k$th user in the $\ell$th cell is represented by $\beta_{j,\ell,k}$. Then, the channel correlation coefficients between the $n$th and the $m$th antennas is given by
\begin{align} \label{eq:One_ring_channel_model}
\left[ {\bf R}_{j,\ell,k}  \right]_{n,m} = \frac{ \beta_{j,\ell,k} }{2\Delta_{j,\ell,k}} \int^{\theta_{j,\ell,k} + \Delta_{j,\ell,k} }_{\theta_{j,\ell,k} - \Delta_{j,\ell,k}} e^{-j \frac{2 \pi}{\lambda} \Psi(\alpha)( {\bf r}_{j,n} - {\bf r}_{j,m}  )} {\rm d}\alpha,
\end{align}
where $\Psi(\alpha) = \left[ \cos(\alpha), \sin(\alpha) \right]$ is the wave vector for a planer wave impinging with the angle of $\alpha$, $\lambda$ is the wavelength, and ${\bf r}_{j,n} = \left[ x_{j,n}, y_{j,n} \right]^{\top}$ is the position vector for the $n$th antenna of the $j$th BS. The eigenvectors and eigenvalues of ${\bf R}_{j,\ell,k}$ contain the spatial correlation information of the channel.

% With the Karhunen-Loeve model, the downlink channel from the $j$th BS to the $k$th user in the $\ell$th cell can be represented by
% \begin{align}
% {\bf h}_{j,\ell,k}^{\sf H}
% \end{align}

\subsection{CSIT Assumption}

We assume a block fading model, in which the downlink channel state ${\bf H}_{\ell,\ell}=\left[{\bf h}_{\ell,\ell,1}, \ldots,{\bf h}_{\ell,\ell,K}  \right] \in \mathcal{C}^{N \times K}$ changes independently over each transmission block, while it keeps a constant within a transmission block. The probability density function is denoted by $f_{\sf H}\left({\bf H}_{\ell,\ell}\right)$. It is assumed that users are able to estimate CSI with sufficiently high accuracy, i.e., perfect CSI at receiver (CSIR). The CSIR can be obtained by a completely standard pilot-sided coherent detector, where a very small amount of downlink pilot symbols are sent to each precoded user data stream, as prescribed today in the LTE and 5G standards. Whereas, the BS is assumed to have limited knowledge of downlink CSIT, i.e., ${\bf \hat H}_{\ell,\ell}=\left[{\bf \hat h}_{\ell,\ell,1}, \ldots,{\bf \hat h}_{\ell,\ell,K}  \right]$. In practice, this limited CSIT knowledge is acquired by quantized feedback in FDD systems \cite{Love_2008,Jindal_2006,Yoo_2006} and by uplink training in TDD systems thanks to the channel reciprocity \cite{Hassibi_Hochwald_2003,Jose_Vishwanath_2011,Yin_Gesbert_2013}. The joint fading process is assumed to be stationary and ergodic with a given first-order joint marginal distribution of $\left({\bf H}_{\ell,\ell},{\bf \hat H}_{\ell,\ell}\right)$ \cite{Caire_Kumar_2007}.

Let ${\bf  \hat h}_{\ell,\ell,k}$ be the minimum mean square error (MMSE) estimate of downlink channel ${\bf h}_{\ell,\ell,k}$. Then, we model imperfect CSIT with error ${\bf e}_{\ell,\ell,k}$ as
\begin{align}
	{\hat {\bf h}}_{\ell,\ell,k} = {\bf h}_{\ell,\ell,k} -  {\bf e}_{\ell,\ell,k}.
\end{align}
Since ${\bf h}_{\ell,\ell,k}$ is assumed to be Gaussian, ${\hat {\bf h}}_{\ell,\ell,k}$ and ${\bf e}_{\ell,\ell,k}$ are jointly Gaussian; thereby, ${\bf e}_{\ell,\ell,k}$ is independent of ${\bf \hat h}_{\ell,\ell,k}$. Then, the distribution of the CSIT error, $f_{{\sf e}}\left({\bf e}_{\ell,\ell,k}\right)$, is characterized by the conditional density function $f_{{\sf h}|{\sf \hat h}}\left({\bf h}_{\ell,\ell,k}|{\bf \hat h}_{\ell,\ell,k}\right)$. From the standard theory of MMSE estimation of Gaussian random vectors, we have that the estimation error ${\bf e}_{\ell,\ell,k}$ is Gaussian with mean zero and given covariance matrix ${\bf \Phi}_{\ell,\ell,k}$. By the orthogonality principle, this covariance matrix is given by ${\bf \Phi}_{\ell,\ell,k} = {\bf R}_{\ell,\ell,k} - {\bf {\hat R}}_{\ell,\ell,k}$, where ${\bf {\hat R}}_{\ell,\ell,k}$ is the covariance of ${\bf {\hat h}}_{\ell,\ell,k}$. This imperfect CSIT model is particularly suitable for TDD MU-MIMO systems. Thanks to the channel reciprocity, the downlink channel can be estimated by the orthogonal uplink pilot transmission across multiple cells in the TDD MU-MIMO systems. When the MMSE estimation is applied to estimate the uplink channel ${\bf h}_{\ell,\ell,k}$, the error covariance matrix ${\bf \Phi}_{\ell,\ell,k}$ is obtained as a function of the spatial correlation matrix of the channels, ${\bf  R}_{\ell,\ell,k}$, uplink transmit power $p^{\rm ul}$, and the length of uplink pilot sequence, $\tau^{\rm ul}$. For example, assuming that the same set of mutually orthogonal pilot sequences is reused in each cell, the covariance matrix of the channel estimation error can be represented as  ${\bf \Phi}_{\ell,\ell,k}= {\bf R}_{\ell,\ell,k}-   {\bf  R}_{\ell,\ell,k}  \left(  \sum_{j = 1}^{L}{\bf  R}_{\ell,j,k}  +\frac{\sigma^2}{\tau^{\rm ul}p^{\rm ul}}{\bf I}_{ N}\right)^{-1} {\bf  R}_{\ell,\ell,k}$ \cite{Yin_Gesbert_2013}.

 {\bf Remark 1 (Imperfect CSIT model for FDD MU-MIMO systems):} For the case of FDD MU-MIMO systems, the accuracy of the CSIT error, ${\bf e}_{\ell,\ell,k}$, is mainly determined by the amount of feedback bits to quantize the downlink channel \cite{Jindal_2006,Yoo_2006}. Specifically, let ${\bf \Lambda}_{\ell,\ell,k}$ be a diagonal matrix containing the non-zero eigenvalues of the spatial correlation matrix ${\bf R}_{\ell,\ell,k}$, and ${\bf U}_{\ell,\ell,k}$ be the matrix of the associated eigenvectors. Then, the imperfect CSIT of the BS can be modeled as: 
\begin{align}
{\bf {\hat h}}_{\ell,\ell,k} = {\bf U}_{\ell,\ell,k} {\bf \Lambda}_{\ell,\ell,k}^{\frac{1}{2}} \left(\sqrt{1 - \kappa^2_{\ell,\ell,k} }{\bf g}_{\ell,\ell,k}  + \kappa_{\ell,\ell,k} {\bf v}_{\ell,\ell,k}   \right), 
\end{align}
where ${\bf g}_{\ell,\ell,k}$ and ${\bf v}_{\ell,\ell,k}$ have IID $\mathcal{CN}(0,1)$ entries, and $\kappa_{\ell,\ell,k} \in [0,1]$ indicates the quality of instantaneous CSIT. This imperfect CSIT model is also applicable to our optimization framework.

\subsection{Downlink Ergodic Spectral Efficiency}

We also denote the transmit signal of the $\ell$th BS at time slot $t$ by ${\bf x}_{\ell}[t] \in \mathbb{C}^{N \times 1}$ where $t\in\left[1, T_{\rm c}\right]$. Each BS independently supports the associated $K$ users by treating all other cell interference as an additional noise. The $\ell$th BS sends $K$ independent information symbols $\left\{x_{\ell,1}[t],\ldots,x_{\ell,K}[t]\right\}$ at time slot $t$ using linear precoding vectors $\{ {\bf f}_{\ell,1},\ldots, {\bf f}_{\ell,K}\}$. Each information symbol is assumed to be a Gaussian signal with zero mean and variance $P$, i.e., $x_{\ell,k}[t] \sim \mathcal{CN}(0,P)$, where $P$ is the total transmit power per cell. The linear precoding vectors at the $\ell$th BS are constructed as a function of imperfect CSIT $\left\{{\hat {\bf h}}_{\ell,\ell,1}, \ldots, {\hat {\bf h}}_{\ell,\ell,K}\right\}$, which  causes IUI in the downlink transmission. We also denote the transmit signal of the $\ell$th BS at time slot $t$ by ${\bf x}_{\ell}[t] \in \mathbb{C}^{N \times 1}$ where $t\in\left[1, T_{\rm c}\right]$. Then, the transmit signal of the $\ell$th BS at time slot $t$ is 
\begin{align}
	{\bf x}_{\ell}[t] = \sum_{k = 1}^{K} {\bf f}_{\ell,k} x_{\ell,k}[t],  
\end{align}
with the transmission power constraint per cell, i.e., $\sum_{k = 1}^{K} \| {\bf f}_{\ell,k} \|^2_2 = 1$. Note that the transmit power for the data symbol of the $k$th user in the $\ell$th cell is computed as $P\| {\bf f}_{\ell,k} \|^2_2$. The received signal of the $k$th user in the $\ell$th cell is
    	\begin{align}
     		y_{\ell,k}[t]=  {{\bf h}^{{\sf H}}_{\ell,\ell,k}}{\bf f}_{\ell,k}x_{\ell,k}[t] +\underbrace{\sum_{i\neq k}^K  {{\bf h}^{{\sf H}}_{\ell,\ell,k}}{\bf f}_{\ell,i}x_{\ell,i}[t]}_{\sf IUI} + \underbrace{\sum_{j\neq \ell}^L\sum_{i=1}^K  {\bf h}_{j,\ell,k}^{{\sf H}}{\bf f}_{j,i}x_{j,i}[t]}_{\sf ICI} +z_{\ell,k}[t], \label{eq:received}
    	\end{align} 
    	where $z_{\ell,k}[t] \sim \mathcal{CN}(0,\sigma^2)$ is the complex Gaussian noise with zero mean and variance $\sigma^2=\mathbb{E}\left[ \| z_{\ell,k}[t]\|^2 \right]$. Assuming that the $k$th user in the $\ell$th cell has the perfect knowledge of the precoded downlink channel state information, i.e.,  ${\bf h}_{\ell,\ell,k}^{\sf H}{\bf f}_{\ell,k}$, the ergodic achievable spectral efficiency of the $k$th user in the $\ell$th cell is
    	\begin{align}
    		{\bar R}_{\ell,k}=\mathbb{E}\left[\log_2\left(1+{\rm SINR}_{\ell,k}\right)\right],\label{eq:ergodic_SE}
    		\end{align}
    	 	where \begin{align}
    	 	{\rm SINR}_{\ell,k}=\frac{ |{{\bf h}^{{\sf H}}_{\ell,\ell,k}}{\bf f}_{\ell,k}|^2}{\sum_{i\neq k}^K   |{{\bf h}^{{\sf H}}_{\ell,\ell,k}}{\bf f}_{\ell,i}|^2 +  \sum_{j\neq \ell}^L\sum_{i=1}^K  |{\bf h}_{j,\ell,k}^{{\sf H}}{\bf f}_{j,i}|^2+\frac{{\sigma}^2}{P}},
    	\end{align} 
    	where $P$ is the total transmit power per cell. In \eqref{eq:ergodic_SE}, the expectations are taken over all fading terms including the desired, inter-user-interference, and inter-cell-interference links. 

\section{Problem Formulation}
In this section, we present a maximization problem for a weighted-sum of spectral efficiencies in multi-cell MU-MIMO systems when imperfect CSIT is available.  

\subsection{Instantaneous Spectral Efficiency Maximization Problem with Imperfect CSIT}
When designing user selection, power allocation, and precoding strategies using limited CSIT knowledge $\left\{{\hat {\bf h}}_{\ell,\ell,1}, \ldots, {\hat {\bf h}}_{\ell,\ell,K}\right\}$, it is impossible for the BS to exactly know the instantaneous downlink spectral efficiency per user, i.e., $\log_2\left(1+{\rm SINR}_{\ell,k}\right)$. This fact possibly makes the BS overestimate the instantaneous spectral efficiency, which leads to the transmission at an undecodable rate. Although the BS cannot perfectly predict the instantaneous rates, it can compute the instantaneous spectral efficiency per downlink user using imperfect CSIT, i.e., ${\bf \hat H}_{\ell,\ell}=\left\{{\hat {\bf h}}_{\ell,\ell,1}, \ldots, {\hat {\bf h}}_{\ell,\ell,K}\right\}$ by taking the expectation with respective to the CSIT error distribution $f_{{\sf h}|{\sf \hat h}}\left({\bf h}_{\ell,\ell,k}|{\bf \hat h}_{\ell,\ell,k}\right)$, which is defined as
\begin{align}
	R_{\ell,k}\left({\bf \hat H}_{\ell,\ell}\right) = \mathbb{E}_{ {\bf H}_{\ell,\ell} | {\bf {\hat H}}_{\ell,\ell}  } \left[\log_2\left(1+{\rm SINR}_{\ell,k}\right)\mid {\bf \hat H}_{\ell,\ell} \right].\label{eq:avg_SE}
\end{align}
 This average spectral efficiency is an instantaneous rate which captures the average rate over the CSIT error distribution when an estimate of CSIT is given. We refer to this as the instantaneous spectral efficiency with imperfect CSIT. Using this instantaneous spectral efficiency with imperfect CSIT, the BS is possible to calculate the ergodic spectral efficiency by taking the expectation over the estimated fading process of its own cell, namely, 
	 \begin{align} \label{eq:avg_SE_ESE}
	 {\tilde R}_{\ell,k} =\mathbb{E}_{ {\bf {\hat H}}_{\ell,\ell} }\left[ R_{\ell,k}\left({\bf \hat H}_{\ell,\ell}\right)  \right]. 
	\end{align}
	 This ergodic spectral efficiency differs from the ergodic spectral efficiency ${\bar R}_{\ell,k}$ in \eqref{eq:ergodic_SE}, because in  \eqref{eq:avg_SE_ESE} the average is taken over the estimated fading channels of its own cell, i.e., local and imperfect CSIT, by treating the inter-cell-interference as additional noise. Whereas, in \eqref{eq:ergodic_SE}, the averages are taken over all fading terms including the inter-cell-interference. Nevertheless, when $L$ BSs perform full-cooperation using global yet imperfect CSIT,  the ergodic spectral efficiencies defined in \eqref{eq:ergodic_SE} and \eqref{eq:avg_SE_ESE} become identical. Therefore, when solving joint user selection, power allocation, and precoding design problem, we focus on maximizing the weighted-sum of the instantaneous spectral efficiency using limited CSIT under the total power constraint in every fading state. 
	
	Let $\mathcal{K}_{\ell}=\{1,2,\ldots, K\}$ be a set of downlink user indices in the $\ell$th cell. Then, the power set of $\mathcal{K}_{\ell}$, which contains all collections of subsets of $\mathcal{K}_{\ell}$, is denoted by $\mathcal{P}(\mathcal{K}_{\ell})$ where its cardinality is $|\mathcal{P}(\mathcal{K}_{\ell})|=2^K$. We define $\mathcal{S}_{\ell, j}$ be the $j$th element of the power set $\mathcal{P}(\mathcal{K}_{\ell})$, where $j\in\left\{1,2,\ldots, 2^K\right\}$. Thus, $\mathcal{S}_{\ell, j}$ is the subset of $\mathcal{P}(\mathcal{K}_{\ell})$. We also define the $i$th element of the subset $\mathcal{S}_{\ell, j}$ by $\pi_{\mathcal{S}_{\ell,j}}(i)\in \{1,2,\ldots, K\}$. For example, when $K=2$, excepting the empty-set $\emptyset$, we have $2^K-1$ subsets in $\mathcal{P}(\mathcal{K}_{\ell})$, i.e., $\mathcal{S}_{\ell,1}=\{1\}$, $\mathcal{S}_{\ell,2}=\{2\}$, and $\mathcal{S}_{\ell,3}=\{1,2\}$. In addition, $\pi_{\mathcal{S}_{\ell,2}}(1)=2$ and $\pi_{\mathcal{S}_{\ell,3}}(1)=1$. Using these notations, the joint user-selection and the weighted sum-spectral efficiency maximization problem is formulated as
\begin{align}
& \underset{{\bf f}_{\ell,\pi_{\mathcal{S}_{\ell,j}}(1)},\ldots,{\bf f}_{\ell,\pi_{\mathcal{S}_{\ell,j}}(|\mathcal{S}_{\ell,j}|)}}{\text{arg~max}}~~ \underset{\mathcal{S}_{\ell,j} \in \mathcal{P}(\mathcal{K}_{\ell})\setminus \emptyset}{\text{arg~max}}
~~   \sum_{i=1}^{|\mathcal{S}_{\ell,j}|}w_{\ell,\pi_{\mathcal{S}_{\ell,j}}(i)}R_{\ell,\pi_{\mathcal{S}_{\ell,j}}(i)}\left({\bf \hat H}_{\ell,\ell}(\mathcal{S}_{\ell,j})\right)  \nonumber\\
& \text{subject to}
~~  \sum_{i=1}^{|\mathcal{S}_{\ell,j}|}\|{\bf f}_{\ell,\pi_{\mathcal{S}_{\ell,j}}(i)}\|_2^2\leq 1, \label{eq:Optimization_general}
\end{align} 
 where ${\bf \hat H}_{\ell,\ell}(\mathcal{S}_{\ell,j})$ is a row-reduced channel matrix that contains only the imperfect CSIT of the selected users in $\mathcal{S}_{\ell,j}$. In addition, $w_{\ell,\pi_{\mathcal{S}_{\ell,j}}(i)}$ is the weight allocated to the $\pi_{\mathcal{S}_{\ell,j}}(i)$th user in the $\ell$th cell. This weight parameter controls between the fairness of the users and the sum-spectral efficiency. For example, to maximize the lower bound of the sum-spectral efficiency, we can set $w_{\ell,\pi_{\mathcal{S}_{\ell,j}}(i)}=1$. In addition, $w_{\ell,\pi_{\mathcal{S}_{\ell,j}}(i)}$ can also be chosen as the inverse of the single-user capacity or using the proportional-fairness criterion to improve the fairness of the downlink user rates in the cell. To find a global optimal solution, we need to find the optimal precoding vectors and associated power allocation solutions for all possible scheduled user subsets $\mathcal{S}_{\ell,j}$. Since the size of the search space exponentially increases with the number of users $2^K-1$\footnote{When $K> N$, the size of the search space can be reduced to $K \choose N$.} and its computation complexity becomes prohibitive especially for large values of $K$.  In addition, for the optimally chosen user sets, finding the optimal precoding vectors and associated power allocation solutions is a non-convex problem.

 %Therefore, low-complexity suboptimal algorithms have been proposed in the literature in order to maximize the throughput solving (11) in two phases (class-B approach): first by finding a set  of quasiorthogonal users (combinatorial search) and second by allocating resources to such a set (convex optimization) 

 \subsection{A Lower Bound of Instantaneous Spectral Efficiency with Imperfect CSIT}   	
    	
We assume that the BS has imperfect CSIT of the users in the cell, i.e., $\left\{{{\bf \hat h}_{1}},\ldots, {{\bf \hat h}_{K}}\right\}$. Then, by using the fact that ${\hat {\bf h}}_{\ell,\ell,k}={\bf h}_{\ell,\ell,k} - {\bf e}_{\ell,\ell,k}$, the received signal of the $k$th user in the $\ell$th BS in \eqref{eq:received} is equivalently rewritten as
    	\begin{align}
    	  		y_{\ell,k}[t]&=\!  {{\bf \hat h}^{{\sf H}}_{\ell,\ell,k}}{\bf f}_{\ell,k}x_{\ell,k}[t] \!+\!\sum_{i\neq k}^K  {{\bf \hat h}^{{\sf H}}_{\ell,\ell,k}}{\bf f}_{\ell,i}x_{\ell,i}[t]\!+\!{{\bf e}^{{\sf H}}_{\ell,\ell,k}}\sum_{i=1}^K {\bf f}_{\ell,i}x_{\ell,i}[t]  \nonumber\\
     		&+ \sum_{j\neq \ell}^L\sum_{i=1}^K  {\bf h}_{j,\ell,k}^{{\sf H}}{\bf f}_{j,i}x_{j,i}[t] +z_{\ell,k}[t].
    	\end{align}
Computing the exact mutual information $I\left(\{x_{\ell,k}[t]\} ; \{y_{\ell,k}[t]\} \right)$ is a very challenging task because the interference plus noise is non-Gaussian due to imperfect CSIT.  In this case, GMI facilitates to estimate a lower bound of the instantaneous spectral efficiency with imperfect CSIT. Using this GMI, the estimate of the instantaneous spectral efficiency for the $k$th downlink user with limited knowledge of the channels, $\left\{{\bf \hat h}_{\ell,\ell,1},\ldots, {\bf \hat h}_{\ell,\ell,K}\right\}$, is
    	\begin{align}
    	 	 & R_{\ell,k}\left({\bf \hat H}_{\ell,\ell}\right)  \!\geq \! \log_2\!\left(\!1\!+\!\frac{ |{{\bf \hat h}^{{\sf H}}_{\ell,\ell,k}}{\bf f}_{\ell,k}|^2}{\sum_{i\neq k}^K\!|{\bf \hat h}^{{\sf H}}_{\ell,\ell,k}{\bf f}_{\ell,i}|^2 \!+\!\sum_{i=1}^K\! {\bf f}_{\ell,i}^{\sf H}{\bf \Phi}_{\ell,\ell,k}{\bf f}_{\ell,i}  \!+\! \frac{{\tilde \sigma}_{\ell,k}^2}{P} }\right),\label{eq:lower_bound}
\end{align}
where the interference leakage power due to imperfect CSIT is computed as 
\begin{align}
	{\bf f}_{\ell,k}^{\sf H}{\bf \Phi}_{\ell,\ell,k}{\bf f}_{\ell,k}P =\mathbb{E}\left[\left|{{\bf e}^{{\sf H}}_{\ell,\ell,k}}{\bf f}_{\ell,k}x_{\ell,k}[t]\right|^2\right]
\end{align} 
and ${\tilde \sigma}_{\ell,k}^2$ is the effective noise variance of the $k$th user in the $\ell$th cell, which includes the sum of the inter-cell-interference and noise power, i.e., ${\tilde \sigma}_{\ell,k}^2=\mathbb{E}\left[\sum_{j\neq \ell}^L\sum_{i=1}^K  \left|{\bf h}_{j,\ell,k}^{{\sf H}}{\bf f}_{j,i}x_{j,i}[t]\right|^2\right]+\sigma^2$.

Using \eqref{eq:lower_bound}, the lower bound of the weighted-sum of the instantaneous spectral efficiency with imperfect CSIT is reformulated as  	 	
\begin{align}
 \sum_{k=1}^Kw_{\ell,k}R_{\ell,k}\left({\bf \hat H}_{\ell,\ell}\right)& \geq \!\sum_{k=1}^Kw_{\ell,k}\!\log_2\!\left(\!1\!+\!\frac{ |{{\bf \hat h}^{{\sf H}}_{\ell,\ell,k}}{\bf f}_{\ell,k}|^2}{\sum_{i\neq k}^K\!|{\bf \hat h}^{{\sf H}}_{\ell,\ell,k}{\bf f}_{\ell,i}|^2 \!+\!\sum_{i=1}^K{\bf f}_{\ell,i}^{\sf H}{\bf \Phi}_{\ell,\ell,k}{\bf f}_{\ell,i} \!+\! \frac{{\tilde \sigma}^2_{\ell,k}}{P} }\!\right) \nonumber \\
& =\!\sum_{k=1}^Kw_{\ell,k}\log_2\!\left(\! \frac{\sum_{k=1}^K|{\bf \hat h}^{{\sf H}}_{\ell,\ell,k}{\bf f}_{\ell,i}|^2 \!+\!\sum_{i=1}^K{\bf f}_{\ell,i}^{\sf H}{\bf \Phi}_{\ell,\ell,k}{\bf f}_{\ell,i}\!+\! \frac{{\tilde \sigma}^2_{\ell,k}}{P}}{\sum_{i\neq k}^K|{\bf \hat h}^{{\sf H}}_{\ell,\ell,k}{\bf f}_{\ell,i}|^2 \!+\!\sum_{k=1}^K{\bf f}_{\ell,i}^{\sf H}{\bf \Phi}_{\ell,\ell,k}{\bf f}_{\ell,i} \!+\! \frac{{\tilde \sigma}^2_{\ell,k}}{P} }\!\right)\nonumber \\
&\!=\!\log_2\!\left(\!\prod_{k=1}^K\!\left[\!\frac{\sum_{i=1}^K\! {\bf f}_{\ell,i}^{\sf H}\left({{\bf \hat h}_{\ell,\ell,k}}{{\bf \hat h}^{{\sf H}}_{\ell,\ell,k}} \!\!+\! {\bf \Phi}_{\ell,\ell,k} \right){\bf f}_{\ell,i} \!+\! \frac{{\tilde \sigma}^2_{\ell,k}}{P}}{\sum_{i\neq k}^K  {\bf f}_{\ell,i}^{\sf H} {{\bf \hat h}_{\ell,\ell,k}}{{\bf \hat h}^{{\sf H}}_{\ell,\ell,k}} {\bf f}_{\ell,i}\!+\!\sum_{i=1}^K \! {\bf f}_{\ell,i}^{\sf H} {\bf \Phi}_{\ell,\ell,k} {\bf f}_{\ell,i}\!+\!\! \frac{{\tilde \sigma}^2_{\ell,k}}{P} }\!\right]^{\!w_{\ell,k}}\!\right).
\end{align}
Therefore, the maximization problem for the weighted sum of the instantaneous spectral efficiency with imperfect CSIT is equivalently written as
\begin{align}
\!& \underset{{\bf f}_{\ell,1},\ldots,{\bf f}_{\ell,K}}{\text{arg~max}}
\!&  & \!\!  \!\! \!\! \!\!\!\!\prod_{k=1}^K\!\left[\!\frac{\sum_{i=1}^K\! {\bf f}_{\ell,i}^{\sf H}\left({{\bf \hat h}_{\ell,\ell,k}}{{\bf \hat h}^{{\sf H}}_{\ell,\ell,k}} \!\!+\! {\bf \Phi}_{\ell,\ell,k} \right){\bf f}_{\ell,i} \!+\! \frac{{\tilde \sigma}^2_{\ell,k}}{P}}{\sum_{i\neq k}^K  {\bf f}_{\ell,i}^{\sf H} {{\bf \hat h}_{\ell,\ell,k}}{{\bf \hat h}^{{\sf H}}_{\ell,\ell,k}} {\bf f}_{\ell,i}\!+\!\sum_{i=1}^K \! {\bf f}_{\ell,i}^{\sf H} {\bf \Phi}_{\ell,\ell,k} {\bf f}_{\ell,i}\!+\!\! \frac{{\tilde \sigma}^2_{\ell,k}}{P} }\!\!\right]^{\!w_{\ell,k}}  \nonumber\\
& \text{subject to}
& & \sum_{i=1}^K\|{\bf f}_{\ell,i}\|_2^2\leq 1. \label{eq:Optimization1}
\end{align} 
Notice that the optimization problem in \eqref{eq:Optimization1} is still a non-convex problem. Unlike the original problem in \eqref{eq:Optimization_general}, the optimization problem in \eqref{eq:Optimization1} maximizes the sum-spectral efficiency under the premise that all $K$ users are scheduled, even in the case of $K>N$. By considering all optimization variables, we aims at finding a joint solution of a set of scheduled users,  precoding vectors, and the power allocations in a computationally efficient manner. In the sequel, we present a low complexity algorithm that solves the optimization problem in \eqref{eq:Optimization1} with the guarantee of the first-order optimality.

\section{ Joint User-Selection, Precoding, and Power Allocation Algorithm}

In this section, we present an algorithm that maximizes a lower bound of the weighted-sum of the instantaneous spectral efficiency with imperfect CSIT in multi-cell MU-MIMO systems. The key aspect of the proposed algorithm is to jointly find 1) a set of downlink users, 2) power allocated information symbols, and 3) precoding vectors carrying information symbols.

\subsection{The Proposed Optimization Approach}
The key idea of the proposed method is to reformulate the optimization problem in \eqref{eq:Optimization1} into a product form of Rayleigh quotients. To accomplish this, we first define a large precoding vector used for the $\ell$th BS as ${\bf f}_{\ell}=\left[{\bf f}_{\ell,1}^{\top},{\bf f}_{\ell,2}^{\top},\ldots,{\bf f}_{\ell,K}^{\top}\right]^{\top}\in\mathbb{C}^{NK\times 1}$. Using this precoding vector, we reformulate the numerator of the objective function in \eqref{eq:Optimization1} with the effective channel matrix ${\bf A}_{\ell,\ell,k}\in\mathbb{C}^{NK\times NK}$ as
\begin{align}
	\sum_{i=1}^K\! {\bf f}_{\ell,i}^{\sf H}\left({{\bf \hat h}_{\ell,\ell,k}}{{\bf \hat h}^{{\sf H}}_{\ell,\ell,k}} \!\!+\! {\bf \Phi}_{\ell,\ell,k} \right){\bf f}_{\ell,i} \!+\! \frac{{\tilde \sigma}^2_{\ell,k}}{P} = {\bf f}_{\ell}^{\sf H}{\bf A}_{\ell,\ell,k}{\bf f}_{\ell},
\end{align}
where ${\bf A}_{\ell,\ell,k}\in\mathbb{C}^{NK\times NK}$ is a block diagonal and positive definite matrix defined as
 \begin{align}
& \!\!{\bf A}_{\ell,\ell,k}=\!\!\nonumber \\
&\!\!\!\small\begin{bmatrix}
      {{\bf \hat h}_{\ell,\ell,k}}{{\bf \hat h}^{{\sf H}}_{\ell,\ell,k}} \!\!\!+\!\! {\bf \Phi}_{\ell,\ell,k}  \!\!&\!\! {\bf 0}  \!\!&\!\! {\bf 0}    \!\!&\!\! \dots  \ \!\!&\!\! {\bf 0} \\     
    \vdots  \!\!&\!\!  \ddots  \!\!&\!\!  \vdots \!&\!    \ddots  \!\!&\!\!  \vdots \\
    {\bf 0}  \!\!&\!\!  \vdots  \!\!&\!\!  {{\bf \hat h}_{\ell,\ell,k}}{{\bf \hat h}^{{\sf H}}_{\ell,\ell,k}}\!\!\!+\!\! {\bf \Phi}_{\ell,\ell,k}    \!\!&\!\!  \dots   \!\!&\!\! {\bf 0} \\
    \vdots  \!\!&\!\!  \vdots  \!\!&\!\! \vdots \!\!&\!\!    \ddots  \!\!&\!\! \vdots \\
    {\bf 0} \!\!&\!\! {\bf 0}  \!\!&\!\!  {\bf 0}  \!\!&\!\!   \dots  \!\!&\!\!  {{\bf \hat h}_{\ell,\ell,k}}{{\bf \hat h}^{{\sf H}}_{\ell,\ell,k}} \!\!\!+\!\! {\bf \Phi}_{\ell,\ell,k}  \!\!
\end{bmatrix}  +\frac{{\tilde \sigma}^2_{\ell,k}}{P}{\bf I}_{NK}.\label{eq:Amat}
\end{align}
Similarly, the denominator of the objective function in \eqref{eq:Optimization1} is rewritten with the effective channel matrix ${\bf B}_{\ell,\ell,k}\in\mathbb{C}^{NK\times NK}$ as
\begin{align}
	\sum_{i\neq k}^K  {\bf f}_{\ell,i}^{\sf H} {{\bf \hat h}_{\ell,\ell,k}}{{\bf \hat h}^{{\sf H}}_{\ell,\ell,k}} {\bf f}_{\ell,i}\!+\!\sum_{i=1}^K \! {\bf f}_{\ell,i}^{\sf H} {\bf \Phi}_{\ell,\ell,i} {\bf f}_{\ell,i} \!+\! \frac{{\tilde \sigma}^2_{\ell,k}}{P} = {\bf f}_{\ell}^{\sf H}{\bf B}_{\ell,\ell,k}{\bf f}_{\ell},
\end{align}
where ${\bf B}_{\ell,\ell,k}$ is the positive definite matrix defined as
\begin{align}
		{\bf B}_{\ell,\ell,k} &= {\bf A}_{\ell,\ell,k}- \begin{bmatrix}
   0 &0 &0 & \dots  & 0 \\
     \vdots & \vdots & \vdots & \ddots & \vdots \\
    0 & \cdots &  {{\bf \hat h}_{\ell,\ell,k}}{{\bf \hat h}^{{\sf H}}_{\ell,\ell,k}} & \dots  & 0 \\
    \vdots & \vdots & \vdots & \ddots & \vdots \\
    0 & 0 & 0 & \dots  & 0\end{bmatrix}. \label{eq:Bmat}
\end{align}
As a result, the optimization problem in \eqref{eq:Optimization1} is equivalently rewritten as a product form of Rayleigh quotients, i.e.,
\begin{align}
  &\underset{{\bf f}_{\ell} \in \mathbb{C}^{NK\times 1}}{\text{arg~max}} 
~~~\prod_{k=1}^K\left[\frac{ {\bf f}_{\ell}^{\sf H} {\bf A}_{\ell,\ell,k}  {\bf f}_{\ell} }{ {\bf f}_{\ell}^{\sf H} {\bf B}_{\ell,\ell,k}  {\bf f}_{\ell} }\right]^{w_{\ell,k}}  \nonumber \\
& \text{subject to} ~~~  \|{\bf f}_{\ell}\|_2^2\leq 1. \label{eq:Reformulated}
\end{align} 
Unfortunately, the objective function in \eqref{eq:Reformulated} is neither convex nor concave function; thereby, finding the global optimal solution of the problem in \eqref{eq:Reformulated} is a very challenging task. Instead, we are able to obtain a suboptimal solution by finding a solution that satisfies the first order necessary Karush-Kuhn-Tucker (KKT) condition. Let the objective function in \eqref{eq:Reformulated} be
\begin{align}
	\lambda({\bf f}_{\ell})&=\prod_{k=1}^K\left[\frac{ {\bf f}_{\ell}^{\sf H} {\bf A}_{\ell,\ell,k}  {\bf f}_{\ell} }{ {\bf f}_{\ell}^{\sf H} {\bf B}_{\ell,\ell,k}  {\bf f}_{\ell} }\right]^{w_{\ell,k}}.
\end{align}
Because for any nonzero $\alpha$, $\lambda({\bf f}_{\ell})=\lambda(\alpha{\bf f}_{\ell})$, we first derive the first order KKT condition of the optimization problem in \eqref{eq:Reformulated} by ignoring the norm constraint on the precoding vector, i.e., $\|{\bf f}_{\ell}\|_2^2\leq1$. Then, the sub-optimal solution vector that satisfies the KKT condition is projected to the unit sphere to be a feasible point. The following lemma shows the first order KKT condition of the unconstrained optimization problem in \eqref{eq:Reformulated}.
\begin{lem}
	The first order KKT condition, i.e., $\frac{\partial \lambda({\bf f}_{\ell})}{\partial {\bf f}_{\ell}^{\sf H}}=0$ satisfies 
  		\begin{align}
  			{\bf \bar A}_{\ell,\ell}({\bf f}_{\ell}) {\bf f}_{\ell} =\lambda({\bf f}_{\ell}) {\bf \bar B}_{\ell,\ell}({\bf f}_{\ell}) {\bf f}_{\ell}, \label{eq:KKT}
  		\end{align}
where 
\begin{align}
	{\bf \bar A}_{\ell,\ell}({\bf f}_{\ell})&=\sum_{i=1}^K w_{\ell,i}\left( {\bf f}_{\ell}^{{\sf H}}{\bf A}_{\ell,\ell,i}{\bf f}_{\ell}\right)^{\!w_{\ell,i}\!-\!1}\left(\prod_{k\neq i}^K {\bf f}_{\ell}^{{\sf H}}{\bf A}_{\ell,\ell,k}{\bf f}_{\ell}\right) {\bf A}_{\ell,\ell,i}, \nonumber\\
	{\bf \bar B}_{\ell,\ell}({\bf f}_{\ell})&= \sum_{i=1}^Kw_{\ell,i}\left(  {\bf f}_{\ell}^{{\sf H}}{\bf B}_{\ell,\ell,i}{\bf f}_{\ell}\right)^{\!w_{\ell,i}\!-\!1}\left(\prod_{k\neq i}^K {\bf f}_{\ell}^{{\sf H}}{\bf B}_{\ell,\ell,k}{\bf f}_{\ell}\right) {\bf B}_{\ell,\ell,i}.
	\end{align}
\end{lem}
 
\begin{proof}
	The proof is direct from the derivative property. Let $f_k({\bf x})=\left({\bf x}^{\sf H}{\bf A}_k{\bf x}\right)^{w_k}$ and $g_k({\bf x})=\left({\bf x}^{\sf H}{\bf B}_k{\bf x}\right)^{w_k}$ be functions with respective to ${\bf x}\in\mathbb{C}^N$, where ${\bf A}_k$ and ${\bf B}_k$ are positive definite matrices for all $k\in\mathcal{K}$. Let $\lambda({\bf x})= \frac{\prod_{k=1}^Kf_k({\bf x})}{\prod_{k=1}^Kg_k({\bf x})} =\frac{f({\bf x}) }{g({\bf x})}$. By the definition, the first-order KKT condition satisfies the following condition:
	\begin{align}
		\frac{ \lambda({\bf x})}{\nabla {\bf x}^{\sf H}}=\frac{ \frac{\prod_{k=1}^K f_k({\bf x}) }{ \nabla {\bf x}^{\sf H} }\prod_{k=1}^Kg_k({\bf x}) - \frac{\prod_{k=1}^K g_k({\bf x}) }{\nabla {\bf x}^{\sf H} }\prod_{k=1}^K f_k({\bf x})  }{ \left(\prod_{k=1}^Kg_k({\bf x})\right)^2}=0, 	\label{eq:KKT_condition}
		\end{align}
where \small ${\frac{\prod_{k=1}^K f_k({\bf x})}{\nabla {\bf x}^{\sf H}}=\left[\sum_{i=1}^Kw_k \left({\bf x}^{\sf H} {\bf A}_i{\bf x}\right)^{w_k-1}\prod_{k \neq i}^K \left({\bf x}^{\sf H} {\bf A}_k{\bf x}\right) {\bf A}_i \right]{\bf x}}$ and $\frac{ \prod_{k=1}^Kg_k({\bf x})}{\nabla {\bf x}^{\sf H}}=\left[\sum_{i=1}^Kw_k \left({\bf x}^{\sf H} {\bf B}_i{\bf x}\right)^{w_k-1} \prod_{k \neq i}^K \left({\bf x}^{\sf H} {\bf B}_k{\bf x}\right) {\bf B}_i
\right]{\bf x}$.\normalsize 

The condition in \eqref{eq:KKT_condition} simplifies to
\begin{align}
	\frac{ \prod_{k=1}^Kf_k({\bf x}) }{ \nabla {\bf x}^{\sf H} }&=\frac{\prod_{k=1}^Kf_k({\bf x})}{\prod_{k=1}^Kg_k({\bf x})} \frac{ \prod_{k=1}^Kg_k({\bf x}) }{\nabla {\bf x}^{\sf H} }  =\lambda({\bf x})\frac{ \prod_{k=1}^Kg_k({\bf x}) }{\nabla {\bf x}^{\sf H} }.
	\end{align}
	This completes the proof. 
\end{proof}
%\frac{\frac{\frac{\frac{\prod_{k=1}^Kf_k({\bf x})}{\nabla {\bf x}^{\sf H}}\prod_{k=1}^Kg_k({\bf x})}}{\left(\prod_{k=1}^Kg_k({\bf x})\right)^2}

We provide an intuitive interpretation of the first-order KKT condition derived in Lemma 1 with the lens through a generalized eigenvalue problem. Finding a solution that satisfies the first-order KKT condition in \eqref{eq:KKT_condition}, in general, is challenging, because it is difficult to efficiently solve the set of $NK$ nonlinear equations with $NK$ unknown variables, especially when the number of antennas at the BS and the number of users are large enough, i.e., $NK\geq 1000$. To overcome this difficulty, the proposed approach is to treating the value of the objective function $\lambda({\bf f}_{\ell})=\prod_{k=1}^K\left[\frac{ {\bf f}_{\ell}^{\sf H} {\bf A}_{\ell,\ell,k}  {\bf f}_{\ell} }{ {\bf f}_{\ell}^{\sf H} {\bf B}_{\ell,\ell,k}  {\bf f}_{\ell} }\right]^{w_{\ell,k}} $ as an eigenvalue of the matrix $\left[{\bf \bar B}_{\ell,\ell}({{\bf f}}_{\ell})\right]^{-1} {\bf \bar A}_{\ell,\ell}({{\bf f}}_{\ell}) $ and the unknown vector ${\bf f}_{\ell} $ as the eigenvector corresponding to the eigenvalue. Then, the optimal solution of our optimization problem is obtained, provided that the first eigenvector of $\left[{\bf \bar B}_{\ell,\ell}({{\bf f}}_{\ell})\right]^{-1} {\bf \bar A}_{\ell,\ell}({{\bf f}}_{\ell}) $ is found. However, our optimization problem differs from a conventional generalized eigenvalue problem by the fact that the system matrix $\left[{\bf \bar B}_{\ell,\ell}({{\bf f}}_{\ell})\right]^{-1} {\bf \bar A}_{\ell,\ell}({{\bf f}}_{\ell}) $ is also a function of unknown ${\bf f}_{\ell}$. This makes the exact solution of the original problem difficult. Nevertheless, in the next subsection, we present a computationally efficient algorithm that converges to a feasible point that satisfies the first-order KKT optimality condition, although it may not be the global optimal solution.

\subsection{Generalized Power Iteration Precoding}
To overcome this computational difficulty, in this subsection, we propose a novel generalized power iteration precoding (GPIP) algorithm. The proposed GPIP algorithm is a simple yet computationally-efficient algorithm that indentifies the principal eigenvector of $\left[{\bf \bar B}_{\ell,\ell}({{\bf f}}_{\ell})\right]^{-1} {\bf \bar A}_{\ell,\ell}({{\bf f}}_{\ell}) $. 

Before presenting the proposed algorithm, we first need to understand the structure of the two matrices ${\bf \bar A}_{\ell,\ell}({\bf f}_{\ell})$ and  ${\bf \bar B}_{\ell,\ell}({\bf f}_{\ell})$. The matrix ${\bf \bar A}_{\ell,\ell}({\bf f}_{\ell})\in \mathbb{C}^{NK\times NK}$ is a linear combination of $K$ matrices, i.e., 
\begin{align}
	{\bf \bar A}_{\ell,\ell}({\bf f}_{\ell})=\sum_{i=1}^Kc_{\ell,\ell,i}({\bf f}_{\ell}){\bf A}_{\ell,\ell,i}, \label{eq:Astructure}
\end{align}
where $c_{\ell,\ell,i}({\bf f}_{\ell})$ is the positive weight of ${\bf A}_{\ell,\ell,i}$ defined as $c_{\ell,\ell,i}({\bf f}_{\ell})=w_{\ell,i}\left( {\bf f}_{\ell}^{{\sf H}}{\bf A}_{\ell,\ell,i}{\bf f}_{\ell}\right)^{\!w_{\ell,i}\!-\!1}\left(\prod_{k\neq i}^K {\bf f}_{\ell}^{{\sf H}}{\bf A}_{\ell,\ell,k}{\bf f}_{\ell}\right)$. Since the coefficient is a function of variable vector ${\bf f}_{\ell}$, the matrix ${\bf \bar A}_{\ell,\ell}({\bf f}_{\ell})$ is also function of ${\bf f}_{\ell}$. Similarly, the matrix ${\bf \bar B}_{\ell,\ell}({\bf f}_{\ell})\in \mathbb{C}^{NK\times NK}$ is a linear combination of $K$ matrices, i.e., 
\begin{align}
	{\bf \bar B}_{\ell,\ell}({\bf f}_{\ell})=\sum_{i=1}^Kd_{\ell,\ell,i}({\bf f}_{\ell}){\bf B}_{\ell,\ell,i}, \label{eq:Bstructure}
\end{align}
where the positive weight $d_{\ell,\ell,i}({\bf f}_{\ell})=\sum_{i=1}^Kw_{\ell,i}\left(  {\bf f}_{\ell}^{{\sf H}}{\bf B}_{\ell,\ell,i}{\bf f}_{\ell}\right)^{\!w_{\ell,i}\!-\!1}\left(\prod_{k\neq i}^K {\bf f}_{\ell}^{{\sf H}}{\bf B}_{\ell,\ell,k}{\bf f}_{\ell}\right)$. In addition, ${\bf \bar A}_{\ell,\ell}({\bf f}_{\ell})\in \mathbb{C}^{NK\times NK}$ and ${\bf \bar B}_{\ell,\ell}({\bf f}_{\ell})\in \mathbb{C}^{NK\times NK}$ are block diagonal and positive-definite matrices, because ${\bf A}_{\ell,\ell,i}$ and ${\bf B}_{\ell,\ell,i}$ are also block diagonal and positive-definite matrices. These matrix properties will be exploited when designing a computationally-efficient algorithm in the sequel.

%The following theorem is the major result of this paper.

 The proposed GPIP algorithm starts with the initial solution of ${\bf f}_{\ell}^0$, which can be a simple maximum ratio transmission (MRT) precoding solution. In the $m$th iteration, for a given precoding vector ${\bf f}_{\ell}^{(m-1)}$, the algorithm computes the system matrix $\left[{\bf \bar B}_{\ell,\ell}\left({\bf f}_{\ell}^{(m-1)}\right)\right]^{-1}{\bf \bar A}_{\ell,\ell}\left({\bf f}_{\ell}^{(m-1)}\right) \in \mathbb{C}^{NK\times NK}$. Once this system matrix is given, the algorithm updates ${\bf f}_{\ell}^{(m)}$ by multiplying the previously updated precoding vector ${\bf f}_{\ell}^{(m-1)}$ into the system matrix $\left[{\bf \bar B}_{\ell,\ell}\left({\bf f}_{\ell}^{(m-1)}\right)\right]^{-1}{\bf \bar A}_{\ell,\ell}\left({\bf f}_{\ell}^{(m-1)}\right)$, i.e., ${\bf f}_{\ell}^{(m)}=\left[{\bf \bar B}_{\ell,\ell}\left({\bf f}_{\ell}^{(m-1)}\right)\right]^{-1}{\bf \bar A}_{\ell,\ell}\left({\bf f}_{\ell}^{(m-1)}\right){\bf f}_{\ell}^{(m-1)}$. To satisfy the unit power constraint, the updated precoding solution is normalized to its norm, i.e., ${\bf f}_{\ell}^{(m)}=\frac{{\bf f}_{\ell}^{(m)}}{\|{\bf f}_{\ell}^{(m)}\|_2}$. If the distance between the updated and the previous precoding solutions is lager than a tolerance level, the algorithm goes to Step 2; otherwise, the algorithm ends. The proposed GPIP is summarized in Table \ref{tab:GPI_AlG}. 
\begin {table}[]
\caption {Generalized Power Iteration Precoding (GPIP) Algorithm} \vspace{-0.5cm}\label{tab:GPI_AlG} 
  	 \begin{center}
  \begin{tabular}{ l | c }
    \hline\hline
    Step 1 & Initialize ${\bf f}_{\ell}^0$ (MRT)  \\ \hline
        Step 2 & In the $m$-th iteration,  \\ \hline
     & Compute $\left[{\bf \bar B}_{\ell,\ell}\left({\bf f}_{\ell}^{(m-1)}\right)\right]^{-1}{\bf \bar A}_{\ell,\ell}\left({\bf f}_{\ell}^{(m-1)}\right)$ \\ \hline     
     & ${\bf f}_{\ell}^{(m)}:=\left[{\bf \bar B}_{\ell,\ell}\left({\bf f}_{\ell}^{(m-1)}\right)\right]^{-1}{\bf \bar A}_{\ell,\ell}\left({\bf f}_{\ell}^{(m-1)}\right){\bf f}_{\ell}^{(m-1)}$ \\ \hline
          & ${\bf f}_{\ell}^{(m)}:=\frac{{\bf f}_{\ell}^{(m)}}{\|{\bf f}_{\ell}^{(m)}\|_2}$ \\ \hline
    Step 3 & Iterates until $\|{\bf f}_{\ell}^{(m-1)}-{\bf f}_{\ell}^{(m)}\|_2\leq \epsilon$\\
    \hline
   \hline
  \end{tabular}
\end{center}
\end {table}

	 {\bf Remark 2 (Generalized power iteration precoding):} GPIP was originally developed in our prior work \cite{Lee_2008} in the context of the two-way relay channel when the relay has multiple antennas. Lemma 1 generalizes the first-order optimality condition in \cite{Lee_2008} by incorporating different weights for each Rayleigh quotient. One interesting observation is that the proposed GPIP is applicable for a different class of problems, particularly for the joint design of user-selection, precoding, and power allocation in the MU-MIMO systems. 
 
\subsection{Algorithm Complexity Analysis}
We provide an analysis for the computational complexity of the proposed GPIP algorithm, which is relevant to the effectiveness of the algorithm for practical use, especially when $N$ and $K$ are sufficiently large, i.e., a massive MIMO setting, $NK\geq 1000$. 

Note that  ${\bf \bar A}_{\ell,\ell,k}  \in \mathbb{C}^{KN \times KN}$ and ${\bf \bar B}_{\ell,\ell,k}  \in \mathbb{C}^{KN \times KN}$ are block diagonal matrices whose sub-block size is $N$-by-$N$.  As can be seen in \eqref{eq:Bmat}, it is observed that the $i$th sub-block matrix of ${\bf \bar B}_{\ell,\ell,k}$ is the form of $ {{\bf \hat h}_{\ell,\ell,k}}{{\bf \hat h}^{{\sf H}}_{\ell,\ell,k}}+ {\bf \Phi}_{\ell,\ell,k}+\frac{{\tilde \sigma}^2_{\ell,k}}{P}{\bf I}_{N} $ for $i\in\mathcal{K}/\{k\}$, while the $k$th sub-block matrix  is the form of ${\bf \Phi}_{\ell,\ell,k}+\frac{{\tilde \sigma}^2_{\ell,k}}{P}{\bf I}_{N}$. Hence, we independently perform the inverse of $K$ sub-matrices, i.e.,  $ {{\bf \hat h}_{\ell,\ell,k}}{{\bf \hat h}^{{\sf H}}_{\ell,\ell,k}} \!\!+\! {\bf \Phi}_{\ell,\ell,k}+\frac{{\tilde \sigma}^2_{\ell,k}}{P}{\bf I}_{N} $ for $k\in\mathcal{K}/\{k\}$ in order to compute $\left[{\bf \bar B}_{\ell,\ell,k}\left({\bf f}_{\ell}^{(m-1)}\right)\right]^{-1}$.  Recall that 
\begin{align}
	{\bf \bar B}_{\ell,\ell}\left({\bf f}_{\ell}^{(m-1)}\right)=\sum_{i=1}^Kd_{\ell,\ell,i}\left({\bf f}_{\ell}^{(m-1)}\right){\bf B}_{\ell,\ell,i},
	\end{align}
Since each sub-block matrix of ${\bf B}_{\ell,\ell,i}$ is symmetric, we need the computational complexity order of $\mathcal{O}( \frac{1}{3}N^3)$ using a Cholesky factorization when obtaining the inverse of each sub-matrix.  As a result, we can compute $\left[{\bf \bar B}_{\ell,\ell,k}\left({\bf f}_{\ell}^{(m-1)}\right)\right]^{-1}$ in a divide-and-conquer manner with the computational complexity order of $\mathcal{O}\left( K\frac{1}{3}N^3\right)$.

Similarly, thanks to the block diagonal structure, we can compute $\left[{\bf \bar B}_{\ell,\ell,k}\left({\bf f}_{\ell}^{(m-1)}\right)\right]^{-1}{\bf \bar A}_{\ell,\ell,k}\left({\bf f}_{\ell}^{(m-1)}\right)$ in a divide-and-conquer manner with the computational complexity order of $\mathcal{O}(KN^2)$. Consequently, we need a total computation complexity order of $\mathcal{O}\left(JK\frac{1}{3}N^3\right)$, where $J$ is the number of iterations. The number of required iteration for the convergence with $\epsilon = 0.1$ is typically less than 4, which will be validated via simulations. 

\begin {table} 
\caption {Computational Complexity of the GPIP Algorithm} \label{tab:GPI_Complexity}  \vspace{-0.5cm}
\begin{center}
  \begin{tabular}{ l | c }
    \hline
      Initial solution & $\mathcal{O}\left(KN\right)$ flops \\ & (for MRT)    \\ \hline
       Computation of ${\bf \bar A}_{\ell,\ell,k}\left({\bf f}_{\ell}^{(m-1)}\right)$ & $\mathcal{O}\left(KN^2\right)$ flops\\ \hline   
      Computation of ${\bf \bar B}_{\ell,\ell,k}\left({\bf f}_{\ell}^{(m-1)}\right)$ & $\mathcal{O}\left(KN^2\right)$ flops\\ \hline
     Computation of $\left[{\bf \bar B}_{\ell,\ell,k}\left({\bf f}_{\ell}^{(m-1)}\right)\right]^{-1}{\bf \bar A}_{\ell,\ell,k}\left({\bf f}_{\ell}^{(m-1)}\right)$ & $\mathcal{O}\left(K\frac{2}{3}N^3\right)$ flops\\ \hline
    \hline
  \end{tabular}
\end{center} \vspace{-1cm}
 \end{table}

 \subsection{No Channel Covariance Matrix Information}
 In this subsection, we consider a practically important case in which the spatial correlation matrices are not available at the BSs. This restriction is common in practice, because estimating and tracking the spatial channel correlation matrix ${\bf R}_{\ell,\ell,k}$ of all users per cell is very challenging in MU-MIMO systems equipped with a large array of antennas \cite{Emil_Debbah_2016,Haghighatshoar_Caire_2017}. In this case, one simple approach is to treat the spatial channel covariance matrix as an identity matrix when deigning precoding vectors. With no spatial channel correlation matrix information, we use an estimate of the covariance matrix for the channel estimation error as $
 	{\bf \hat \Phi}_{\ell,\ell,k}=\alpha_{\ell,\ell, k}{\bf I}_N$ where $\alpha_{\ell,\ell,k}= {\beta}_{\ell,\ell,k}\left(1-\frac{{\beta}_{\ell,\ell,k}}{\sum_{j=1}^L{\beta}_{\ell,j,k} +\frac{\sigma^2}{\tau^{\rm ul}p^{\rm ul}} }\right)$, which causes the mismatch effect of the channel error covariance matrix. With imperfect channel covariance matrix ${\bf \hat \Phi}_{\ell,\ell,k}$, the loss in the spectral efficiency can occur. Due to the diagonal matrix structure of ${\bf \hat \Phi}_{\ell,\ell,k}$, however, we can reduce the computational complexity of the proposed power iteration precoding significantly using the Sherman-Morrison formula, which replaces matrix inversion operation to vector multiplication operation.

 	Recall that, from the definitions in \eqref{eq:Bmat} and \eqref{eq:Bstructure}, ${\bf \bar B}_{\ell,\ell}\left({\bf f}_{\ell}\right)\in \mathbb{C}^{NK \times NK}$ is a block diagonal matrix, i.e., 	\begin{align}
		{\bf \bar B}_{\ell,\ell}\left({\bf f}_{\ell}\right) &\!=\!\small{ \begin{bmatrix}
   {\bf \tilde B}^{(K)}_{\ell,\ell,1}\left({\bf f}_{\ell}\right) &0 &0 & \dots  & 0 \\
      \vdots &\ddots & \vdots & \ddots & \vdots \\
     0 & \cdots &{\bf \tilde B}^{(K)}_{\ell,\ell,k}\left({\bf f}_{\ell}\right)& \dots  & 0 \\
     \vdots & \vdots & \vdots & \ddots & \vdots \\
     0 & 0 & 0 & \dots  & {\bf \tilde B}^{(K)}_{\ell,\ell,K}\left({\bf f}_{\ell}\right)\end{bmatrix}}, \label{eq:Bmat22}
 \end{align}
 where the $j$th sub-block for $j\in\mathcal{K}/\{k\}$ is
 \begin{align}
 {\bf \tilde B}^{(K)}_{\ell,\ell,j}\left({\bf f}_{\ell}\right)=	\sum_{i=1}^Kd_{\ell,\ell, i}({\bf f}_{\ell}){\bf \hat h}_{\ell,\ell,i}{\bf \hat h}_{\ell,\ell,i}^{{\sf H}} + \delta_{\ell,\ell,j}{\bf I}_N \in \mathbb{C}^{N\times N},
 \end{align}
 with $\delta_{\ell,\ell,j}=\sum_{i=1}^Kd_{\ell,\ell, i}({\bf f}_{\ell})\left( \alpha_{\ell,\ell,j}+\frac{{\tilde \sigma}^2_{\ell,i}}{P} \right)$. The $k$th sub-block is
 \begin{align}
 {\bf \tilde B}^{(K)}_{\ell,\ell,k}\left({\bf f}_{\ell}\right)= \delta_{\ell,\ell,k}{\bf I}_N \in \mathbb{C}^{N\times N},
 \end{align}

  We can compute the inverse of this sub-block matrix in a recursive manner by successively applying the Sherman-Morrison formula $\left({\bf A}+{\bf u}{\bf v}^{{\sf H}}\right)^{-1}={\bf A}^{-1}-\frac{{\bf A}^{-1}{\bf u}{\bf v}^{{\sf H}} {\bf A}^{-1}}{1+ {\bf v}^{{\sf H}}{\bf A}^{-1}{\bf u}}$ \cite{Golub_1996}. Let 
 \begin{align}
 	{\bf \tilde B}^{(k)}_{\ell,\ell,j}\left({\bf f}_{\ell}\right) = \delta_{\ell,\ell}{\bf I}_{N}+ \sum_{i=1}^kd_{\ell,\ell, i}({\bf f}_{\ell}){\bf \hat h}_{\ell,\ell,i}{\bf \hat h}_{\ell,\ell,i}^{{\sf H}}.
 	\end{align}
 The inverse of ${\bf \tilde B}^{(k)}_{\ell,\ell,j}\left({\bf f}_{\ell}\right)$ is recursively computed using the previously obtained inverse of ${\bf \tilde B}^{(k-1)}_{\ell,\ell,j}\left({\bf f}_{\ell}\right)$ as
 \begin{align}
 	\left[\!{\bf \tilde B}^{(k-1)}_{\ell,\ell,j}\left({\bf f}_{\ell}\right)\!\right]^{-1} \!&=\! \left[{\bf \tilde B}^{(k-1)}_{\ell,\ell,j}\left({\bf f}_{\ell}\right)\right]^{-1}  \nonumber \\
 	\!&- \!\frac{\left[{\bf \tilde B}^{(k-1)}_{\ell,\ell,j}\left({\bf f}_{\ell}\right)\right]^{-1} {\bf \hat h}_{\ell,\ell,k} {\bf \hat h}_{\ell,\ell,k}^{{\sf H}}\left[\!{\bf \tilde B}^{(k-1)}_{\ell,\ell,j}\left({\bf f}_{\ell}\right)\!\right]^{-1}}{\frac{1}{d_{\ell,\ell, k}({\bf f}_{\ell})}+{\bf \hat h}_{\ell,\ell,k}^{{\sf H}}\left[\!{\bf \tilde B}^{(k-1)}_{\ell,\ell,j}\left({\bf f}_{\ell}\right)\!\right]^{-1}{\bf \hat h}_{\ell,\ell,k}}, \nonumber
 \end{align}
 where the initial inversion is given by
 \begin{align}
 	\left[{\bf \tilde B}^{(1)}_{\ell,\ell,j}\left({\bf f}_{\ell}\right)\right]^{-1} \!&=\left[  d_{\ell,\ell, 1}({\bf f}_{\ell}){{\bf \hat h}_{\ell,\ell,1}}{{\bf \hat h}^{{\sf H}}_{\ell,\ell,1}} +\delta_{\ell,\ell}{\bf I}_N \right]^{-1}\nonumber \\
 	\!&= \delta_{\ell,\ell}^{-1}\left[{\bf I}_N  - \frac{d_{\ell,\ell, 1}({\bf f}_{\ell}) {{\bf \hat h}_{\ell,\ell,1}}{{\bf \hat h}^{{\sf H}}_{\ell,\ell,1}}}{ \delta_{\ell,\ell}+d_{\ell,\ell, 1}({\bf f}_{\ell}) {{\bf \hat h}^{{\sf H}}_{\ell,\ell,1}}{{\bf \hat h}_{\ell,\ell,1}}}\right]\!.
 \end{align}
 Consequently, the inverse of sub-matrix ${\bf \tilde B}^{(k)}_{\ell,\ell,j}$ is obtained with the computational complexity order of $\mathcal{O}\left(KN^2\right)$. This successive matrix inversion technique significantly reduces the total computational complexity order of the proposed power iteration precoding from $\mathcal{O}\left( J\frac{1}{3}KN^3\right)$ to  $\mathcal{O}\left(JKN^2\right)$ for a large number of antenna systems. This computational complexity deduction is particularly interesting because the proposed power iteration precoding requires even a less computational complexity order compared to that of the ZF precoding, which needs $\mathcal{O}( \frac{2}{3}N^3)$, provided that we properly set the number of iterations $J$ and the number of users $K$ such that $JK<N$.

\subsection{Joint User Selection, Power Allocation, and Precoding}
One key feature of the proposed algorithm is that it jointly finds a set of scheduled users, power allocation, and precoding vectors carrying information symbols regardless of the number of users $K$, and the number of antennas $N$. We elucidate this feature by providing an example of $(N,K)=(2,3)$. 

\vspace{0.1cm}
\noindent
{\bf Example 1:} Suppose the channel vectors of the three users in the $\ell$th cell as
\begin{align}
&	{\bf h}_{\ell,\ell,1}=\begin{bmatrix}	
	0.46+0.56j \\ 
	0.08-0.67j
\end{bmatrix}, 	{\bf h}_{\ell,\ell,2}=\begin{bmatrix}	
	0.04+0.33j \\ 
	0.01+0.365j
\end{bmatrix},   ~{\rm and}~ 
{\bf h}_{\ell,\ell,3}=\begin{bmatrix}	
	-0.0031-0.0025j \\ 
	0.0082-0.0038j
\end{bmatrix}.
\end{align}
With an initial solution of MRT precoding, the proposed algorithm yields precoding vectors for the three users when $\frac{P}{{\tilde \sigma}^2_{\ell,k}}=10$ dB as
\begin{align}
&	{\bf f}_{\ell,1}=\begin{bmatrix}	
	0.3554+0.3492j \\ 
	0.1120-0.4573j
\end{bmatrix}, 	{\bf f}_{\ell,2}=\begin{bmatrix}	
	0.1447+0.4697j \\ 
	-0.0647+0.5332j
\end{bmatrix},  ~{\rm and}~ 
{\bf f}_{\ell,3}=\begin{bmatrix}	
	0.000-0.0001j \\ 
	0.000-0.0001j
\end{bmatrix}.
\end{align}
As can be seen in this example, the precoding solutions obtained by the proposed algorithm show that user 1 and user 2 are selected for the transmission because $\|{\bf f}_{\ell,1}\|_2\geq \epsilon$ and $\|{\bf f}_{\ell,2}\|_2\geq \epsilon$, while it deactivates user 3 by assigning a near zero-vector for ${\bf f}_{\ell,3}$. In addition, the solutions contain the power allocation effect because $\|{\bf f}_{\ell,1}\|_2^2=0.47$ and $\|{\bf f}_{\ell,2}\|_2^2=0.53$. As a result, the proposed algorithm jointly provides a set of scheduled users, power allocation, and precoding vectors carrying information symbols.

\section{Extension to Multi-Cell Cooperative Transmission }
In this section, we extend the proposed algorithm for the joint design of user selection, power allocation, and precoding to the multi-cell cooperative downlink transmission. We first explain the BS cooperation model and then show how to extend the proposed GPIP algorithm in the cooperative transmission scenario for multi-cell MU-MIMO systems.

 \subsection{BS Cooperation Model}
 Let $ C(\leq L)$ be the number of cooperative BSs in a cluster.  The $C$ BSs in the cluster are assumed to be connected via high-speed and error-free backhauls. We assume that the cooperative BSs have global knowledge of imperfect CSIT in the cluster, while no data sharing between the BSs is considered. This assumption is feasible because each BS estimates CSIT between itself and all downlink users using the orthogonal uplink pilot transmission in the cooperative cluster (e.g. $\tau^{\rm ul}\geq CK$) and it is shared through the backhauls in every channel coherence intervals $T_{\rm c}$ to perform the cooperative downlink transmission. Our BS cooperation model differs from cell-free massive MIMO and CoMP joint transmission methods in \cite{Gesbert_Yu_2010,Huh_Caire_2011,Huh_Caire_2012,Ngo_Marzetta_2017} in which all downlink data symbols are shared by the BSs using backhauls. Rather, this cooperative transmission strategy is well-known to as multi-cell coordinated scheduling/beamforming in LTE systems \cite{Lee_Lozano_2014}.
 
Let ${\rm IUI}_{\ell,k}$ and ${\rm ICI}_{\ell,k}$ be the aggregated inter-user-interference and the inter-cell-interference received at the $k$th user in the $\ell$th cell, which are defined as
\begin{align}
	{\rm IUI}_{\ell,k}&=\sum_{i\neq k}^K|{\bf \hat h}^{{\sf H}}_{\ell,\ell,k}{\bf f}_{\ell,i}|^2 +\sum_{k=1}^K\mathbb{E}\left[|{{\bf e}^{{\sf H}}_{\ell,\ell,k}}{\bf f}_{\ell,i}|^2\right] \nonumber\\
	&=\sum_{i\neq k}^K {\bf f}_{\ell,i}^{\sf H}{\bf \hat h}^{{\sf H}}_{\ell,\ell,k}{\bf \hat h}^{{\sf H}}_{\ell,\ell,k}{\bf f}_{\ell,i} +\sum_{k=1}^K{\bf f}_{\ell,i}^{\sf H}{\bf \Phi}_{\ell,\ell,k}{\bf f}_{\ell,i} \label{eq:IUI}
\end{align}
and
\begin{align}
	{\rm ICI}_{\ell,k}&=\sum_{j\neq \ell}^{C}\sum_{i=1}^K|{\bf \hat h}^{{\sf H}}_{j,\ell,k}{\bf f}_{j,i}|^2 +\sum_{j\neq \ell}^{C}\sum_{k=1}^K\mathbb{E}\left[|{{\bf e}^{{\sf H}}_{j,\ell,k}}{\bf f}_{j,i}|^2\right] \nonumber\\
	&=\sum_{j\neq \ell}^{C}\sum_{i=1}^K {\bf f}_{j,i}^{\sf H}{\bf \hat h}^{{\sf H}}_{j,\ell,k}{\bf \hat h}^{{\sf H}}_{j,\ell,k}{\bf f}_{j,i} +\sum_{j\neq \ell}^{C}\sum_{i=1}^K{\bf f}_{j,i}^{\sf H} {\bf \Phi}_{j,\ell,k}{\bf f}_{j,i}. \label{eq:ICI}
\end{align}
Then, the achievable spectral efficiency for the $k$th downlink user in the $\ell$th cell with global and imperfect CSIT is 
 \begin{align}
    	 	{R}_{\ell,k}\left({\bf \hat H}_{1,1},\ldots, {\bf \hat H}_{C,C}\right)\!=\!\log_2\!\left(\!1\!+\!\frac{ |{{\bf \hat h}^{{\sf H}}_{\ell,\ell,k}}{\bf f}_{\ell,k}|^2}{{\rm IUI}_{\ell,k}+{\rm ICI}_{\ell,k} \!+\! \frac{{\bar \sigma}^2_{\ell,k}}{P} }\right),
\end{align}
where $\bar{\sigma}^2_{\ell,k}$ is the effective noise variance of the $k$th user in the $\ell$th cell, which includes the sum of the out-of-cluster interference and noise power.

\subsection{Cooperative Precoding for Multi-Cell MU-MIMO}
The our precoding design problem is that of maximizing the weighted spectral efficiencies of all downlink users in the cooperative cluster subject to the transmission power constraint of the individual BS. This optimization problem is formulated as
\begin{align}
 \underset{{\bf f}_{1,1},\ldots,{\bf f}_{C,K}}{\text{max}}
 &\sum_{\ell=1}^{C}\sum_{k=1}^Kw_{\ell,k}{R}_{\ell,k}\left({\bf \hat H}_{1,1},\ldots, {\bf \hat H}_{C,C}\right) \nonumber \\ \label{eq:objectfunction_multicell}
\text{subject to}
& \sum_{k=1}^K\|{\bf f}_{\ell,k}\|_2^2\leq 1~~ {\rm for}~~\ell\in\{1,\ldots,C\}.
\end{align} 
Notice that the object function of the optimization problem in \eqref{eq:objectfunction_multicell} is reformulated as 
\begin{align}
 \sum_{\ell=1}^{C}\sum_{k=1}^Kw_{\ell,k}{R}_{\ell,k}\left({\bf \hat H}_{1,1},\ldots, {\bf \hat H}_{{C},C}\right) &= \sum_{\ell=1}^{C}\sum_{k=1}^Kw_{\ell,k}\!\log_2\!\left(\!1\!+\!\frac{ |{{\bf \hat h}^{{\sf H}}_{\ell,\ell,k}}{\bf f}_{\ell,k}|^2}{{\rm IUI}_{\ell,k}+{\rm ICI}_{\ell,k} \!+\! \frac{\bar{\sigma}^2_{\ell,k}}{P} }\right) \nonumber \\
& =\!\log_2\!\left( \left[ \prod_{\ell=1}^{C}\prod_{k=1}^K\frac{ |{{\bf \hat h}^{{\sf H}}_{\ell,\ell,k}}{\bf f}_{\ell,k}|^2 +{\rm IUI}_{\ell,k}+{\rm ICI}_{\ell,k} \!+\! \frac{ \bar{\sigma}^2_{\ell,k}}{P}}{ {\rm IUI}_{\ell,k}+{\rm ICI}_{\ell,k} \!+\! \frac{ \bar{\sigma}^2_{\ell,k}}{P} }\right]^{w_{\ell,k}}\right). \label{eq:multicell_rate}
 \end{align}
Let ${\bf f}_{\ell}=\left[{\bf f}_{\ell,1}^{\top},{\bf f}_{\ell,2}^{\top},\ldots,{\bf f}_{\ell,K}^{\top}\right]^{\top}\in\mathbb{C}^{NK\times 1}$ and ${\bf f}=\left[{\bf f}_{1}^{\top},{\bf f}_{2}^{\top},\ldots,{\bf f}_{C}^{\top}\right]^{\top}\in\mathbb{C}^{NKC\times 1}$. We also redefine the numerator in \eqref{eq:multicell_rate} using \eqref{eq:IUI} and \eqref{eq:ICI} as a quadratic function with respective to ${\bf f}$, namely,
\begin{align}
	& |{{\bf \hat h}^{{\sf H}}_{\ell,\ell,k}}{\bf f}_{\ell,k}|^2 +{\rm IUI}_{\ell,k}+{\rm ICI}_{\ell,k} \!+\! \frac{\bar{\sigma}^2_{\ell,k}}{P} \nonumber \\
	 & = \sum_{i=1}^K {\bf f}_{\ell,i}^{\sf H}{\bf \hat h}^{{\sf H}}_{\ell,\ell,k}{\bf \hat h}^{{\sf H}}_{\ell,\ell,k}{\bf f}_{\ell,i} +\sum_{k=1}^K{\bf f}_{\ell,i}^{\sf H}{\bf \Phi}_{\ell,\ell,k}{\bf f}_{\ell,i} \nonumber \\
	 &+\sum_{j\neq \ell}^{C}\sum_{i=1}^K {\bf f}_{j,i}^{\sf H}{\bf \hat h}^{{\sf H}}_{j,\ell,k}{\bf \hat h}^{{\sf H}}_{j,\ell,k}{\bf f}_{j,i} +\sum_{j\neq \ell}^{C}\sum_{i=1}^K{\bf f}_{j,i}^{\sf H} {\bf \Phi}_{j,\ell,k}{\bf f}_{j,i}+\frac{\bar{\sigma}^2_{\ell,k}}{P} \nonumber\\
	 &={\bf f}{\bf A}_{\ell,k}^{\rm coop}{\bf f},\label{eq:nominator}
\end{align}
where 
    \begin{align}
 &\!\!\!\!\!\!\!\!\!     {\bf A}_{\ell,k}^{\rm coop}= \begin{bmatrix}
      {\bf A}_{1,\ell,k}^{\rm coop}&0 &0   & \dots  & 0 \\     
    \vdots &  \ddots &\vdots  & \ddots & \vdots \\
    0 &   0 &{\bf A}_{\ell,\ell,k}^{\rm coop} & \dots  & 0 \\
    \vdots &  \vdots & \vdots & \ddots & \vdots \\
    0 & 0 &0 &   \dots  &  {\bf A}_{C,\ell,k}^{\rm coop} \nonumber
\end{bmatrix}\in \mathbb{C}^{CNK\times CNK}
\end{align}
is the positive-semidefinite and block diagonal matrix with the size of $LNK$ by $LNK$ whose $j$th sub-block matrix with the size of $NK$-by-$NK$ also has the block diagonal structure as
 \begin{align}
  {\bf A}_{j,\ell,k}^{\rm coop}={\rm diag}\!\left(\!{{\bf \hat h}_{j,\ell,1}}{{\bf \hat h}^{{\sf H}}_{j,\ell,1}}\!\!+\!{\bf \Phi}_{j,\ell,1}\!\!+\!\!\frac{\bar{\sigma}^2_{\ell,k}}{P }{\bf I}_N,\!
  \cdots \!, {{\bf \hat h}_{j,\ell,K}}{{\bf \hat h}^{{\sf H}}_{j,\ell,K}}\!\!+\!\!{\bf \Phi}_{j,\ell,K}\!\!+\!\!\frac{\bar{\sigma}^2_{\ell,k}}{P }{\bf I}_N\!\right)\!.
\end{align}
Similarly, we can define the denominator in \eqref{eq:multicell_rate} using \eqref{eq:IUI} and \eqref{eq:ICI} as a quadratic function with respective to ${\bf f}$, i.e.,
\begin{align}
	&{\rm IUI}_{\ell,k}+{\rm ICI}_{\ell,k} \!+\! \frac{\bar{\sigma}^2_{\ell,k}}{P} \nonumber \\
	 & = \sum_{i\neq k}^K {\bf f}_{\ell,i}^{\sf H}{\bf \hat h}^{{\sf H}}_{\ell,\ell,k}{\bf \hat h}^{{\sf H}}_{\ell,\ell,k}{\bf f}_{\ell,i} +\sum_{k=1}^K{\bf f}_{\ell,i}^{\sf H}{\bf \Phi}_{\ell,\ell,k}{\bf f}_{\ell,i} \nonumber \\
	 &+\sum_{j\neq \ell}^{C}\sum_{i=1}^K {\bf f}_{j,i}^{\sf H}{\bf \hat h}^{{\sf H}}_{j,\ell,k}{\bf \hat h}^{{\sf H}}_{j,\ell,k}{\bf f}_{j,i} +\sum_{j\neq \ell}^{C}\sum_{i=1}^K{\bf f}_{j,i}^{\sf H}{\bf \Phi}_{j,\ell,k}{\bf f}_{j,i}+\frac{\bar{\sigma}^2_{\ell,k}}{P} \nonumber\\
	 &={\bf f}{\bf B}_{\ell,k}^{\rm coop}{\bf f}, \label{eq:denominator}
\end{align}
where     \begin{align}
 &\!\!\!\!\!\!\!\!\!     {\bf B}_{\ell,k}^{\rm coop}= \begin{bmatrix}
      {\bf B}_{1,\ell,k}^{\rm coop}&0 &0   & \dots  & 0 \\     
    \vdots &  \ddots &\vdots  & \ddots & \vdots \\
    0 &   0 &{\bf B}_{\ell,\ell,k}^{\rm coop} & \dots  & 0 \\
    \vdots &  \vdots & \vdots & \ddots & \vdots \\
    0 & 0 &0 &   \dots  &  {\bf B}_{C,\ell,k}^{\rm coop} 
    \end{bmatrix}\in \mathbb{C}^{CNK\times CNK} \nonumber
\end{align}
is also the positive-semidefinite and block diagonal matrix with the size of $CNK$-by-$CNK$. The $j$th sub-block matrix of ${\bf B}_{\ell,k}^{\rm coop}$ has the size of $NK$-by-$NK$, and it has the block diagonal structure defined as 
    \begin{align}
	{\bf B}_{\ell,\ell,k}^{\rm coop} &= {\bf A}_{\ell,\ell,k}^{\rm coop}- \begin{bmatrix}
   0 &0 &0 & \dots  & 0 \\
     \vdots & \vdots & \vdots & \ddots & \vdots \\
    0 & \cdots &  {{\bf \hat h}_{\ell,\ell,k}}{{\bf \hat h}^{{\sf H}}_{\ell,\ell,k}} & \dots  & 0 \\
    \vdots & \vdots & \vdots & \ddots & \vdots \\
    0 & 0 & 0 & \dots  & 0\end{bmatrix}. 
\end{align}
Using this concatenated multi-cell precoding vector and the system matrices defined in \eqref{eq:nominator} and \eqref{eq:denominator}, the multi-cell precoding design problem is represented as the maximization problem of the product of Rayleigh quotients as follows:
  \begin{align}
\underset{{\bf f}\in\mathbb{C}^{CKL\times 1} }{\text{max}}
& \prod_{\ell=1}^{C}\prod_{k=1}^K \left[ \frac{ {\bf f}^{\sf H} {\bf A}_{\ell,k}^{\rm coop} {\bf f} }{ {\bf f}^{\sf H} {\bf B}_{\ell,k}^{\rm coop} {\bf f} }\right]^{w_{\ell,k}}  \\
 \text{subject to}
&~ \|{\bf f}_{\ell}\|_2^2\leq 1~~ {\rm for}~~\ell\in\{1,\ldots,C\}. \label{eq:multicell_coop_opti}
\end{align} 
Since this optimization problem is also non-convex, we find a suboptimal solution by finding the first order optimality condition when relaxing the individual BS power constraint to the sum-power constraint, i.e., $ \|{\bf f}\|_2^2\leq C$.  Similar to the non-cooperative transmission case, since $\lambda^{\rm coop}({\bf f}) =\lambda^{\rm coop}(\alpha{\bf f}) $ for some $\alpha >0$, it is possible to find the first order-KKT condition ignoring the sum-power constraint. The first-order KKT condition of the optimization problem in \eqref{eq:multicell_coop_opti} with the sum-power constraint is given below:
%\begin{lem}
%  		The first order KKT condition, i.e., $\frac{\partial \lambda({\bf f})}{\partial {\bf f}^{\sf H}}=0$ satisfies 
%  		\begin{align}
%  			&\prod_{\ell=1}^L\prod_{k=1}^K \left({\bf f}^{{\sf H}}{\bf A}_{ \ell,k}^{\rm coop}{\bf f} \right)^{w_{\ell,k}}\!\!\underbrace{\left[\sum_{\ell=1}^L\sum_{i=1}^Kw_{\ell,i}\left( {\bf f}^{{\sf H}}{\bf \tilde A}_{j,i}{\bf f}\right)^{\!w_{\ell,i}\!-\!1}\left(\prod_{\ell \neq j}^L\prod_{k\neq i}^K {\bf f}^{{\sf H}}{\bf \tilde A}_{\ell,k}{\bf f}\right) {\bf \tilde A}_{j,i}\right]}_{{\bf \bar A}_{\ell, k}({\bf f} ) }{\bf f}   \nonumber \\
%  			&\!\!\!\!\!=\prod_{\ell=1}^L\prod_{k=1}^K\left( {\bf f}^{{\sf H}}{\bf \tilde B}_{\ell,k}{\bf f} \right)^{w_{\ell,k}}\!\!\underbrace{\left[\sum_{\ell=1}^L\sum_{i=1}^Kw_{\ell,i}\left( {\bf f}^{{\sf H}}{\bf \tilde B}_{j,i}{\bf f}\right)^{\!w_{\ell,i}\!-\!1}\left(\prod_{\ell \neq j}^L\prod_{k\neq i}^K {\bf f}^{{\sf H}}{\bf \tilde B}_{\ell,k}{\bf f}\right) {\bf \tilde B}_{\ell,i}\right]}_{{\bf \bar B}_{\ell,k}({\bf f})} {\bf f}.
%  		\end{align}
%  	\end{lem}

\begin{lem}
	The first order KKT condition, i.e., $\frac{\partial \lambda^{\rm coop}({\bf f})}{\partial {\bf f}^{\sf H}}=0$ with the sum-power constraint $\|{\bf f}\|_2\leq C$ satisfies 
  		\begin{align}
  			{\bf \tilde A}^{\rm coop}({\bf f}) {\bf f}=\lambda^{\rm coop}({\bf f}) {\bf \tilde B}^{\rm coop}({\bf f}) {\bf f}, \label{eq:KKT}
  		\end{align}
where 
\begin{align}
	{\bf \tilde A}^{\rm coop}({\bf f})&\!=\!\! \sum_{\ell=1}^{C}\sum_{k=1}^K\!w_{\ell,k}\left( {\bf f}^{{\sf H}}{\bf  A}^{\rm coop}_{\ell,k}{\bf f}\right)^{\!w_{\ell,k}\!-\!1}\!\left(\prod_{j \neq \ell }^{C}\prod_{i \neq k}^K {\bf f}^{{\sf H}}{\bf  A}^{\rm coop}_{j,i}{\bf f}\right) {\bf A}^{\rm coop}_{\ell,k} , \nonumber\\
	{\bf \tilde B}^{\rm coop}({\bf f})&\!=\!\!  \sum_{\ell=1}^{C}\sum_{k=1}^K\!w_{\ell,k}\left( {\bf f}^{{\sf H}}{\bf  B}^{\rm coop}_{\ell,k}{\bf f}\right)^{\!w_{\ell,k}\!-\!1}\!\left(\prod_{j \neq \ell}^{C}\prod_{i \neq k}^K {\bf f}^{{\sf H}}{\bf  B}_{j,i}^{\rm coop}{\bf f}\right) {\bf  B}_{\ell,k}^{\rm coop}, \nonumber\\
	\lambda^{\rm coop}({\bf f})&=\prod_{\ell=1}^{C}\prod_{k=1}^K \left[ \frac{ {\bf f}^{\sf H} {\bf A}_{\ell,k}^{\rm coop} {\bf f} }{ {\bf f}^{\sf H} {\bf B}_{\ell,k}^{\rm coop} {\bf f} }\right]^{w_{\ell,k}}  . 
\end{align}
\end{lem}
\begin{proof}
	The proof is direct from the proof of Lemma 1.
	\end{proof}

By modifying the GPIP algorithm in Table \ref{tab:GPI_AlG}, we present the multi-cell cooperative precoding algorithm as in Table \ref{tab:GPI_AlG2}. The multi-cell cooperative precoding algorithm is almost identical to the GPIP algorithm in Table \ref{tab:GPI_AlG}. One major difference is the scaling operation in Step 4. Since this algorithm finds the solution without considering the individual power constraint, we need to normalize the solution to satisfy the individual BS power constraint, i.e., $\|{\bf f}_{\ell}\|_2=1$. To accomplish this, we rescale the maximum eigenvector of $\left[{\bf \tilde B}^{\rm coop}({\bf f})\right]^{-1}{\bf \tilde A}^{\rm coop}({\bf f})$ by $\max_{\ell \in \mathcal{L}}\left\{\|{\bf f}^{(m)}_{\ell}\|_2 \right\}$. This normalization guarantees to satisfy the individual BS power constraint in \eqref{eq:multicell_coop_opti}. The multi-cell cooperative precoding algorithm in Table \ref{tab:GPI_AlG2} converges to a stationary point of the optimization problem in \eqref{eq:multicell_coop_opti} with the sum-power constraint. In addition, thank to the systematic block diagonal structure of the matrices ${\bf A}_{\ell,k}^{\rm coop} $ and ${\bf B}_{\ell,k}^{\rm coop} $, one can easily show that the computational complexity of the multi-cell precoding algorithm scales quadratically with the number of cooperative BSs ${C}$, i.e., $\mathcal{O}(JK{C}^2N^2)$.

%
%
%{\bf Remark 2 (Multi-cell scheduling):}  Similar to the non-cooperation case, the proposed algorithm performs multi-cell scheduling by selecting a subset of users among the entire $KL$ downlink users in the cooperative cells so as to maximize the weighted sum-spectral efficiency. Specifically, once we obtain the solution vector ${\bf f}^{\star}\in \mathbb{C}^{LKN}$ using the power iteration precoding for multi-cell precoding, it is possible that some subvectors of ${\bf f}^{\star}\in \mathbb{C}^{LKN}$ converge to near-zero, i.e., ${\bf f}_{\ell,k}\simeq {\bf 0}\in \mathbb{C}^{N}$. This implies that the proposed algorithm does not choose the $k$th user in the $\ell$th cell to be scheduled in order to maximize the weighted sum-spectral efficiency. Whereas, if $\|{\bf f}_{\ell,k}\|_2>\epsilon$, the $\ell$th BS should send the data symbol for the corresponding downlink user. Therefore, the scheduled user set by the proposed algorithm can be defined as
%\begin{align}
%	\mathcal{S} =\left\{(\ell, k) \mid \|{\bf f}_{\ell,k}\|_2>\epsilon 
%	\right\}, 
%\end{align}
%for $\ell \in \mathcal{L}$ and $k\in \mathcal{K}$. 

 \begin {table}[]
\caption {Generalized Power Iteration for Multi-Cell Precoding} \label{tab:GPI_AlG2} 
 \vspace{-0.5cm} 	 \begin{center}
  \begin{tabular}{ l | c }
    \hline\hline
    Step 1 & Initialize ${\bf f}^0$ (MRT)  \\ \hline
        Step 2 & In the $m$-th iteration,  \\ \hline
     & Compute $\left[{\bf \tilde B}^{\rm coop}\left({\bf f}^{(m-1)}\right)\right]^{-1}{\bf \tilde A}^{\rm coop}\left({\bf f}^{(m-1)}\right)$ \\ \hline     
     & ${\bf f}^{(m)}:=\left[{\bf \tilde B}^{\rm coop}\left({\bf f}^{(m-1)}\right)\right]^{-1}{\bf \tilde A}^{\rm coop}\left({\bf f}^{(m-1)}\right){\bf f}^{(m-1)}$ \\ \hline
          & ${\bf f}^{(m)}:=\frac{{\bf f}^{(m)}}{\|{\bf f}^{(m)}\|_2}$ \\ \hline
    Step 3 & Iterates until $\|{\bf f}^{(m-1)}-{\bf f}^{(m)}\|_2\leq \epsilon$\\ \hline
    Step 4 & Rescaling:  ${\bf f}^{(m)}:=\frac{\|{\bf f}^{(m)}\|_2}{\max_{\ell \in \mathcal{L}}\left\{\|{\bf f}^{(m)}_{\ell}\|_2 \right\}}$\\
    \hline
   \hline
  \end{tabular}
\end{center} \vspace{-1cm}
\end {table}

\section{Simulation Results}
In this section, we provide both the link and the system level simulation results to compare the performance of the proposed GPIP with those of the existing precoding methods in downlink multi-cell MU-MIMO systems. We assume that the BS is equipped with uniform circular array with $N$ isotropic antennas in which the antenna elements are equally spaced on a circle of radius. The circle of radius is set to $\lambda D$, where $D = \frac{0.5}{\sqrt{ \left( 1 - \cos(2\pi/N ) \right)^2 + \sin(2\pi/N)^2 }}$ leads to the minimum distance $\lambda/2$ between adjacent antennas.

%The assumptions of link level simulations is that we consider a single cell operation and ignore large scale terms, i.e, $L = 1$ and $\beta_{i,j,k} = 1 , \forall i,j,k$. The assumptions of system-level simulations are summarized in Table 2.
%%%%%%%%%%%%%%%%%%%%%%%%%%%%%%%%%%%%%%%%%%%%%%%%%%%%%%%%%%%%%%%%%%%%%%%%%%%%
\subsection{Link Level Simulations}
We first present the link level simulation results for the single-cell ($L=1$) MU-MIMO system, in which the large-scale fading terms are ignored, i.e., $\beta_{\ell,\ell,k} = 1$, and the small-scale fading terms are generated by using the geometric one-ring scattering model in \eqref{eq:One_ring_channel_model}, i.e., ${\bf h}_{\ell,\ell,k} \sim \mathcal{CN}({\bf 0},{\bf R}_{\ell,\ell,k})$. We assume that the users are uniformly located at an azimuth angle $\theta_{\ell,\ell,k} = 2 \pi k /K$ and angular spread $\Delta_{\ell,\ell,k} = \pi/6$. We set the initial solution of the proposed GPIP as the MRT solution. In addition, we set the tolerance level of $\epsilon=0.01$ for the proposed GPIP in Table I. For link level simulations, we use the uniform weight values, i.e., $w_{\ell,k} = 1$.

We compare the proposed GPIP with the following well-known precoding and user selection algorithms:

\begin{itemize}
 \item ZF-DPC \cite{Caire_Shamai_2003}: this scheme serves as the information-theoretical upper-bound of the downlink sum-spectral efficiency for the MU-MIMO systems when perfect CSIT is available. The water-filling power allocation method is applied; 
	\item SUS-ZF \cite{Yoo_Goldsmith_2005}: this algorithm refers the user selection algorithm based on semi-orthogonal user selection and zero-forcing precoding. The computational complexity of this algorithm is $\mathcal{O}\left(KN^3\right)$;
	
	\item RRZF \cite{Wang2012}: this scheme refers a robust regularized zero-forcing (RRZF) precoding, which can improve the performance of regularized zero-forcing precoding (RZF) with imperfect CSIT. The RRZF precoding solution of the $\ell$th BS is given by
	\begin{align}
    {\bf F}_{\ell}^{\rm RRZF} = {\bf {\hat H}}^{\sf H}_{\ell,\ell}\left( {\bf {\hat H}}_{\ell,\ell}{\bf {\hat H}}^{\sf H}_{\ell,\ell} +\sum_{k = 1}^{K} {\bf \Phi}_{\ell,\ell,k}  +\frac{\sigma^2}{P}{\bf I}_{N} \right)^{-1},
\end{align}
where ${\bf F}_{\ell} = \left[ {\bf f}_{\ell,1}, \ldots,  {\bf f}_{\ell,K} \right] \in \mathbb{C}^{N \times K}$. This RRZF provides the same performance with that of RZF when perfect CSIT is available.
	
	\item Rank-adaptation with ZF: this algorithm selects a set of scheduled users so that it maximizes the sum-spectral efficiency with ZF precoding in a greedy manner. Specifically, it first selects a user who produces the maximum single-user capacity. Then, the algorithm finds the second user so that it can yield the maximum of the sum-spectral efficiency together with the previously selected user when applying ZF precoding. In this manner, it performs the greedy-user selection by adding users, until there is no increase of the sum-spectral efficiency. Notice that the computational complexity of this method is much higher than that of SUS-ZF and the proposed GPIP. This is because this greedy algorithm needs to perform the $r\times N$ matrix inversion $(r-1)$ times in the $r$th iteration for ZF precoding to check whether the sum-spectral efficiency increases or not by adding more users.
\end{itemize}
%%%%%%%%%%%%%%%%%%%%%%%%%%%%%%%%%%%%%%%%%%%%%%%%%%%%%%%%%%
\begin{figure} [t]
	\centering
    \includegraphics[width=12cm]{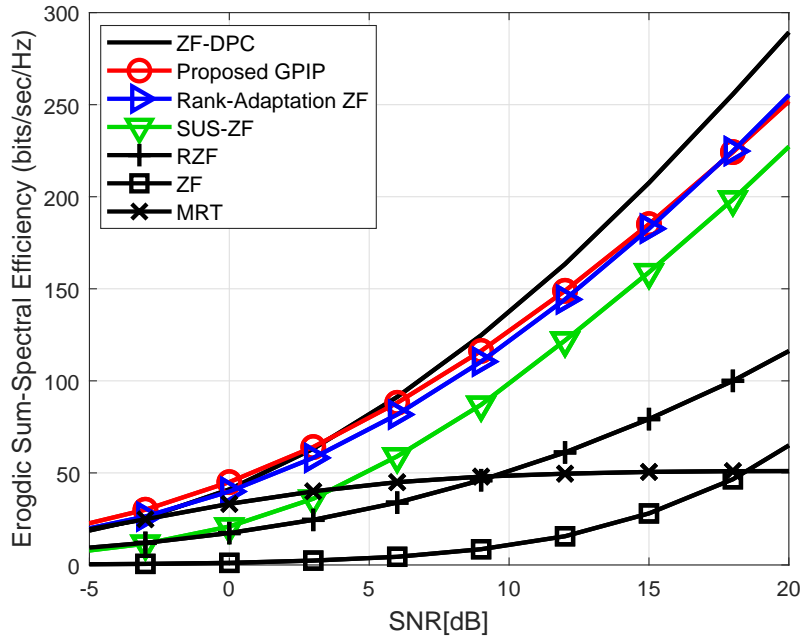}
  \caption{Ergodic sum-spectral efficiency comparisons under perfect CSIT.}\vspace{-0.5cm} \label{Fig_Sim_Link_Perfect_CSIT}
\end{figure}
%%%%%%%%%%%%%%%%%%%%%%%%%%%%%%%%%%%%%%%%%%%%%%%%%%%%%%%%%

{\bf Perfect CSIT:} Fig. \ref{Fig_Sim_Link_Perfect_CSIT} shows the achievable sum-spectral efficiencies of the different downlink transmission strategies under the perfect CSIT assumption when $(N,K) = (64,64)$. As shown in Fig. \ref{Fig_Sim_Link_Perfect_CSIT}, the proposed GPIP provides a higher sum-spectral efficiency in all SNR regimes compared to the existing linear precoding methods including ZF, regularized ZF (RZF), and the SUS-ZF algorithm in \cite{Yoo_Goldsmith_2005}. Specifically, the proposed GPIP attains about $3$ dB and $1.5$ dB SNR gains at low and mid SNR regimes compared to the SUS-ZF method. In addition, the proposed GPIP is slightly better than the rank-adaptation ZF method when SNR is below $15$ dB. One remarkable result is that the proposed GPIP achieves a near optimal sum-spectral efficiency attained by ZF-DPC with water-filling power control in the low SNR regime. Nevertheless, the performance gap between the proposed GPIP and ZF-DPC becomes larger as the SNR increases due to the limitation of linear processing. 

% {\bf Perfect CSIT:} Fig. \ref{Fig_Sim_Link_Imperfect_CSIT} shows the achievable sum-spectral efficiencies of the different downlink transmission strategies under the perfect CSIT assumption when $(N,K) = (64,64)$. As shown in Fig. \ref{Fig_Sim_Link_Imperfect_CSIT}, the proposed power iteration precoding provides a higher sum-spectral efficiency in all SNR regimes compared to the existing linear precoding methods including ZF, regularized ZF (RZF), the SUS-ZF algorithm in \cite{Yoo_Goldsmith_2005}, and ZF-DPC in \cite{Caire_Shamai_2003} as the benchmark. One remarkable result is that the proposed power iteration precoding achieves a near optimal sum-spectral efficiency attained by ZF-DPC with water-filling power control in the low SNR regime (0-4 dB). Nevertheless, the performance gap between the proposed one and ZF-DPC becomes larger as the SNR increases due to the limitation of linear processing. Another interesting observation is that GPI-ZF and SUS-ZF achieves almost similar sum-spectral efficiencies. Whereas, when using the proposed power iteration precoding as the joint solution of the precoding, power allocation, and scheduler, the proposed power iteration precoding attains about $3$ dB and $1.5$ dB SNR gains at low and mid SNR regimes compared to the SUS-ZF method. This result implies that the proposed power iteration precoding is able to obtain additional SNR gains by superior power allocation and precoding compared to the SUS-ZF method. 
%%%%%%%%%%%%%%%%%%%%%%%%%%%%%%%%%%%%%%%%%%%%%%%%%%%%%%
\begin{figure} [t]
	\centering
    \includegraphics[width=12cm]{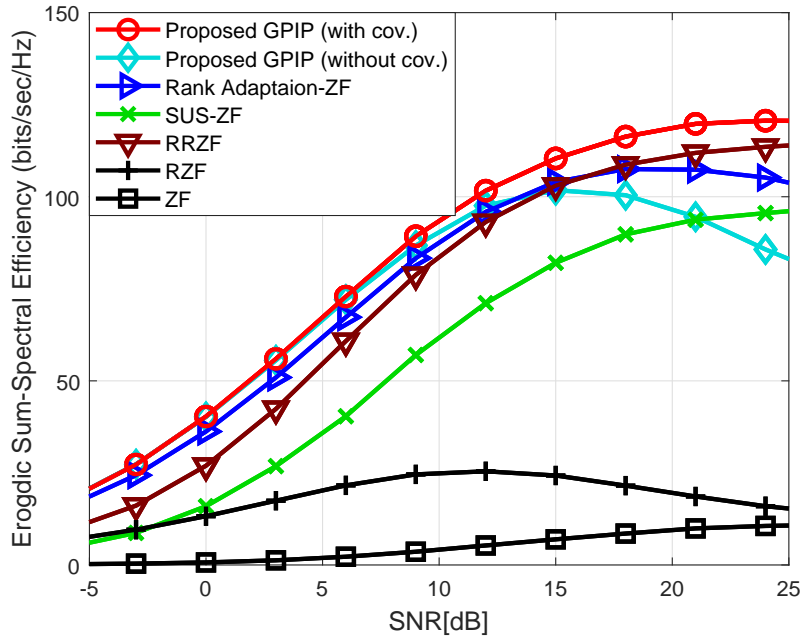}
  \caption{Ergodic sum-spectral efficiency comparisons under imperfect CSIT.} \label{Fig_Sim_Link_Imperfect_CSIT}\vspace{-0.4cm}
\end{figure}
%%%%%%%%%%%%%%%%%%%%%%%%%%%%%%%%%%%%%%%%%%%%%%%%%%%%%%
%Also, the proposed GPI-ZF\footnote{For the proposed GPI-ZF, we set the $\epsilon$ value in \eqref{eq:GPI-ZF} to be $0.001$.} has a similar sum-spectral efficiency with the SUS-ZF. 

{\bf Imperfect CSIT:} To see the robustness to the channel errors, we also compare the achievable sum-spectral efficiencies of the different transmission strategies under  imperfect CSIT assumption. This imperfect CSIT model considered is ${\bf \hat h}_{\ell,\ell,k}={\bf h}_{\ell,\ell,k}+{\bf e}_{\ell,\ell,k}$ where ${\bf e}_{\ell,\ell,k}\sim \mathcal{CN}({\bf 0},0.1\times {\bf I}_N)$, i.e., ${\bf \Phi}_{\ell,\ell,k}=0.1\times{\bf I}_N$. In Fig. \ref{Fig_Sim_Link_Imperfect_CSIT}, we consider two imperfect CSIT scenarios. The first scenario is the case where the BS has perfect knowledge of the error covariance matrix, i.e., ${\bf \hat \Phi}_{\ell,\ell,k}=0.1\times {\bf I}_N$. The second scenario is the case where the BS has no information on the error covariance matrix, i.e., ${\bf \hat \Phi}_{\ell,\ell,k}= {\bf 0}_{N}$. As can be seen in Fig. 2,  when the error covariance matrix information is absent, the sum-spectral efficiencies of RZF and the proposed GPIP decease when the SNRs are above 12dB and 15dB, respectively. These performance degradations are compensated when the error covariance matrix information is available for both RZF and GPIP.  One observation is that the proposed GPIP with the error covariance matrix information provides about 3dB gain over the RRZF method in all SNR ranges.

\subsection{System Level Simulations}
\begin {table}[t]
\caption {System Level Simulation Assumptions} \vspace{-0.5cm}\label{tab:Sys_Assumption} 
  	 \begin{center}
  	 \small
  \begin{tabular}{ l  c }
    \hline\hline
    Parameters & Value  \\ \hline
        Topology & Hexagonal 19 cells  \\ 
         Inter-BS distance & 1000m \\
         Minimum distance btw. MS and BS & 40m \\
         Carrier frequency & 2GHz \\
         Bandwidth & 20MHz \\
         BS transmission power & 40dBm \\
         Spatial channel model & One-ring scattering model in \eqref{eq:One_ring_channel_model}\\
         Path-loss model & Okumaura-Hata model \tablefootnote{With the parameters in Table \ref{tab:Sys_Assumption}, we use the path-loss model as 
$
           L_{\sf p}(D)|_{\rm dB} = 135.1047 + 35.0413 \log_{10} (D),
$
         where $D$ represents the distance in Km between user and BS antennas.}
           \\
         BS/MS height & 32m/1.5m \\
         Shadowing standard deviation & 8dB \\
         Channel estimation & Imperfect \\ \hline
  \end{tabular}
\end{center}\vspace{-1cm}
\end {table}

% \begin{figure} [t]
% 	\centering
%   \includegraphics[width=12cm]{Layout.pdf}
%   \caption{An illustration of the cell-layout for system-level-simulations when $(N,K) = (64,64)$.} \label{Fig.7}
% \end{figure}

In this subsection, we evaluate system level performances of the proposed GPIP, ZF-DPC, ZF, RZF, RRZF, MRT, SUS-ZF, and rank-adaptation with ZF. The set of parameters for the system level simulations is summarized in Table \ref{tab:Sys_Assumption}. In the system level simulations, we use the imperfect CSIT model as ${\bf \hat h}_{\ell,\ell,k}={\bf h}_{\ell,\ell,k}+{\bf e}_{\ell,\ell,k}$ where 
${\bf e}_{\ell,\ell,k}\sim \mathcal{CN}\left({\bf 0}, {\bf R}_{\ell,\ell,k}- {\bf  R}_{\ell,\ell,k}  \left(  \sum_{j \notin \mathcal{C}_{\ell}}{\bf  R}_{\ell,j,k}  +\frac{\sigma^2}{\tau^{\rm ul}p^{\rm ul}}{\bf I}_{ N}\right)^{-1} {\bf  R}_{\ell,\ell,k} \right)$ as in \cite{Yin_Gesbert_2013}. Here, $\mathcal{C}_{\ell}$ denotes the set  BSs, which perform the cooperation with the $\ell$th BS, i.e., $|\mathcal{C}_{\ell}| = C - 1$. This imperfect CSIT model is valid when orthogonal pilot sequences are assigned to the users across the cooperative cells, and they are fully reused in the other cooperative cells. We set the transmission power and the pilot length of each user as $p^{\rm ul}=20$ dBm and $\tau^{\rm ul}=CK$, respectively. In addition, the stopping parameter $\epsilon$ for the proposed GPIP is set to be $0.01$.

\begin{figure} [t]
	\centering
    \includegraphics[width=12cm]{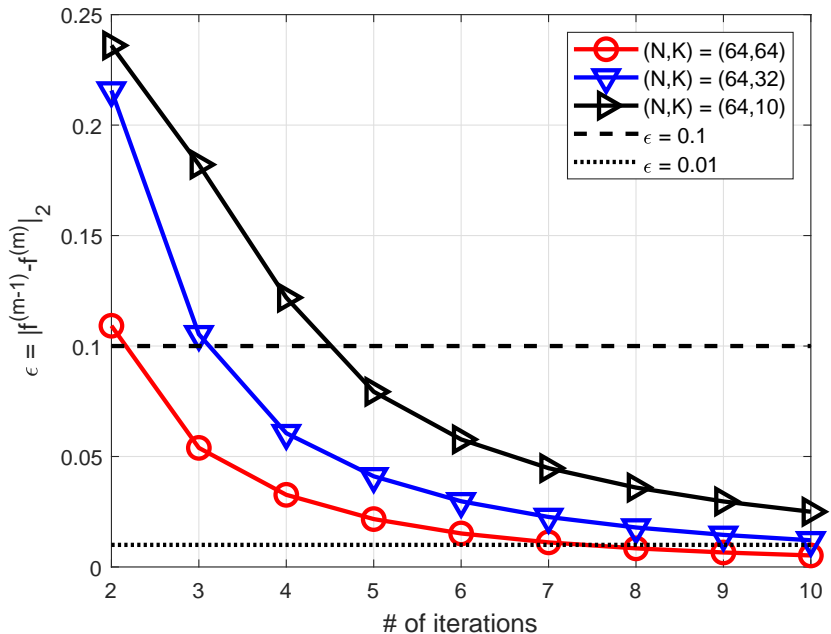}
  \caption{The convergence speed of the proposed GPIP.} \label{Fig_Sim_SLS_Convergence}
\end{figure}

{\bf Convergence:} Fig. \ref{Fig_Sim_SLS_Convergence} shows the convergence speed of the proposed GPIP algorithm in Table \ref{tab:GPI_AlG} for different settings $(N,K) = (64,10)$, $(64,32)$, and $(64,64)$. We observe that the precoding vector ${\bf f}^{(m)}$ quickly converges to the stationary point as the number of iterations increases. The convergence speed of the proposed algorithm depends on the ratio between $N$ and $K$. If we set $\epsilon = 0.1$, which is a stopping condition of the proposed algorithm, four iterations are sufficient to end the algorithm.

\begin{figure} [t]
	\centering
    \includegraphics[width=12cm]{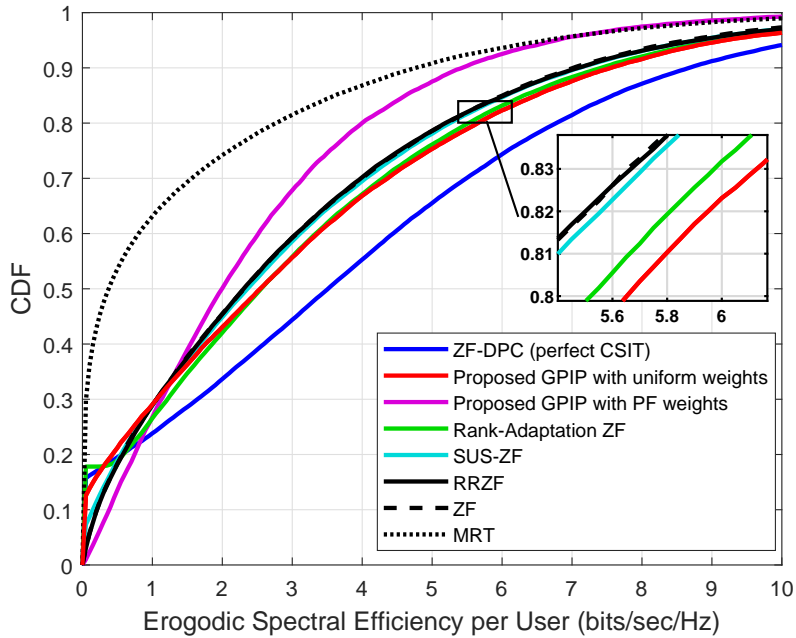}
  \caption{ Rate distributions of different non-cooperative transmission strategies when $(N,K) = (64,10)$.}  \vspace{-0.2cm}\label{Fig_Sim_SLS_Rate_CDF_K_10}
\end{figure}

{\bf Rate distributions of non-cooperative transmission methods:} Fig. \ref{Fig_Sim_SLS_Rate_CDF_K_10} illustrates the distributions of rates per user attained by different transmission methods when $(N,K) = (64,10)$. We evaluate the rate distributions of the proposed GPIP using both the uniform weight and the proportional fair (PF) weight values. As can be seen in Fig. \ref{Fig_Sim_SLS_Rate_CDF_K_10}, the proposed GPIP with the uniform weight shows the spectral efficiency gains compared to the existing precoding methods. It is observed that the proposed GPIP performs the user selection with the water-filling power control effect, because the users with bad-channel conditions are not served by BS. The similar user-selection with the water-filling effect is observed from the rate distribution of ZF-DPC. To improve the fairness of the rate distributions, the proposed algorithm is able to use PF weights. When applying the PF weights, the performance of the cell-edge users significantly improves, while degrading the performance of cell-center users. 

\begin{figure} [t]
	\centering
    \includegraphics[width=12cm]{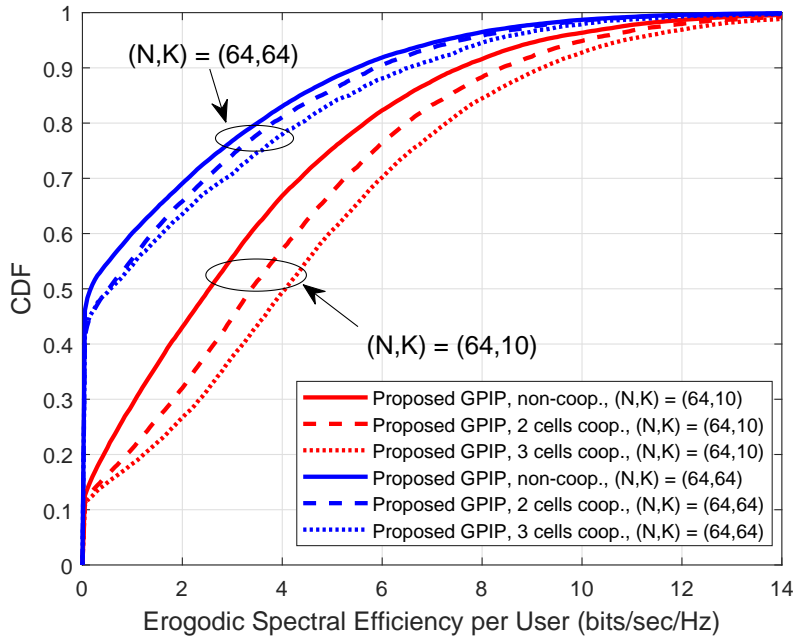}
  \caption{ Rate distributions of multi-cell operations when $(N,K) = (64,10)$ and $(64,64)$. For the proposed GPI precoding, uniform weight values are used.}\label{Fig_Sim_SLS_Rate_CDF_Comp}
\end{figure}

{\bf Rate distributions of multi-cell coordinated scheduling and precoding:} Fig. \ref{Fig_Sim_SLS_Rate_CDF_Comp} illustrates the rate distributions of multi-cell scheduling and precoding based on the proposed method when $(N,K) = (64,10)$ and $(64,64)$. We consider the scenarios where the number of cooperative BSs is two or three, i.e., $\mathcal{C}=\{2,3\}$. Compared to the non-cooperative case, the proposed GPIP for BS cooperation provides the noticeable enhancement of the per-user rate distribution performance, as the number of cooperative BSs increases. In particular, the performance enhancement becomes large when the number of BS antennas is larger than that of users. This fact shows that when the number of BS antennas is much larger than that of users, the proposed GPIP is possible to more cleverly exploit the remaining degrees of freedom of the BS antennas so as to control inter-cell interference. 

\begin{figure} [t]
	\centering
    \includegraphics[width=12cm]{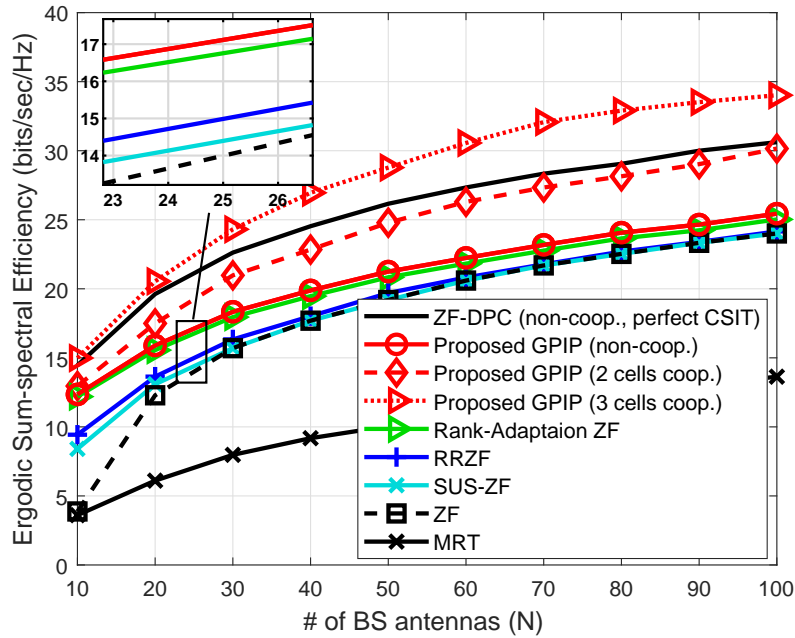}
  \caption{ Ergodic sum-spectral efficiencies when increasing the number of BS antennas for $K = 10$.} \label{Fig_Sim_SLS_Varying_Antennas}
\end{figure}
 %%%%%%%%%%%%%%%%%%%%%%%%%%%%%%%%%%%%%%%%%%%%%%%%%%%%%%%%%%%%%%%%%%%%%%%%%%%%%%%%%%%%%%%%%%%%
 
{\bf Effects of the number BS antennas:} Fig. \ref{Fig_Sim_SLS_Varying_Antennas} shows the ergodic sum-spectral efficiencies achieved by the different transmission methods when increasing the number of BS antennas, while fixing $K=10$. We observe that the proposed GPIP outperforms the existing methods regardless of the number of BS antennas, $N$. Particularly, when $(N,K)=(60,10)$, the proposed GPIP provides about 10 $\%$ spectral efficiency gain over RRZF. In other words, ZF and RRZF need more antennas to achieve the same spectral efficiency of the proposed GPI. For example, to achieve 20 bits/sec/Hz spectral efficiency, the proposed GPIP needs 10 less number of antennas than that needed by RRZF; thereby, the cost of MU-MIMO systems can be reduced by the proposed GPIP. In addition, as can be seen, the proposed cooperative BS transmission method is able to considerably improve the sum-spectral efficiencies as the number of cooperative BSs increases. 

 %%%%%%%%%%%%%%%%%%%%%%%%%%%%%%%%%%%%%%%%%%%%%%%%%%%%%%%%%%%%%%%%%%%%%%%%%%%%%%%%%%%%%%%%%%%%
 \begin{figure} [t]
	\centering
    \includegraphics[width=12cm]{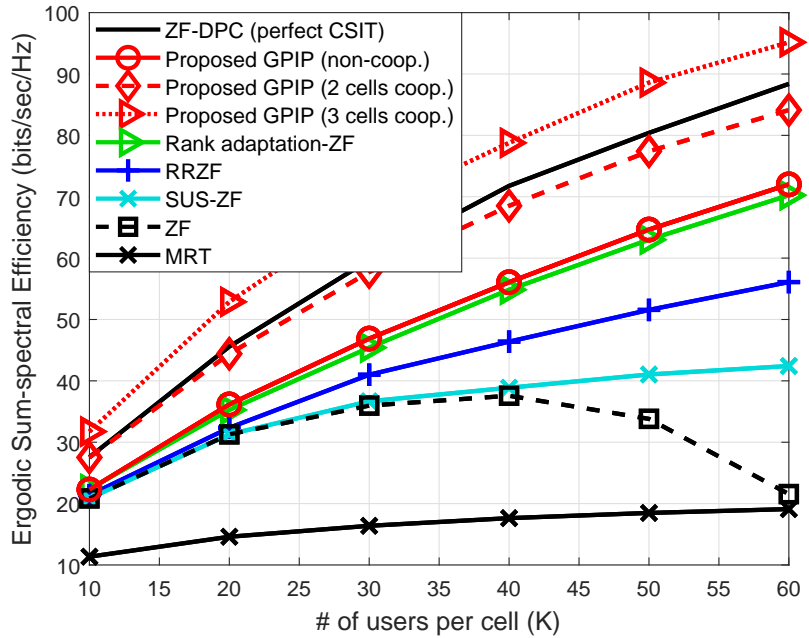}
 \caption{ Ergodic sum-spectral efficiencies when increasing the number of users per cell for  $N = 64$.}\label{Fig_Sim_SLS_Varying_Users}
\end{figure}
 %%%%%%%%%%%%%%%%%%%%%%%%%%%%%%%%%%%%%%%%%%%%%%%%%%%%%%%%%%%%%%%%%%%%%%%%%%%%%%%%%%%%%%%%%%%%
 
{\bf Effects of the number users per cell:} In Fig. \ref{Fig_Sim_SLS_Varying_Users}, we increase the number of downlink users per cell from $10$ to $60$, while fixing $N=64$. Fig. \ref{Fig_Sim_SLS_Varying_Users} shows the ergodic sum-spectral efficiencies per cell. As seen in Fig. \ref{Fig_Sim_SLS_Varying_Users}, the proposed GPI outperforms SUS-ZF, rank-adaptation with ZF, RRZF, and MRT regardless of the number of downlink users, $K$. In addition, although the sum-spectral efficiency obtained by ZF decreases beyond $K\geq 40$, the proposed GPIP improves the sum-spectral efficiency as $K$ increases. This result shows that the user selection is essential in the regime of $N/K<2$, and an additional gain is also attainable by the proposed power allocation and precoding method. In addition, it is observed that the sum-spectral efficiency improves when the number of the cooperative BSs increases. The performance enhancement becomes larger when the number of users per cell increases for the particularly three BS cooperation scenario.This is because the proposed cooperative transmission method is able to obtain a more multi-user diversity gain by performing the multi-cell scheduling as the number of users per cell increases.

\section{Conclusion}
In this paper, we proposed a new linear pre-processing method for downlink multi-cell MU-MIMO systems with imperfect CSIT. The proposed framework was to reformulate the maximization of the WSM problem into the maximization problem for the product of Rayleigh quotients. By proposing a computationally efficient algorithm that ensures the first order optimality condition, we found a sub-optimal solution of the WSM problem. The salient feature of the proposed precoding solution was to find the joint solutions for user selection, power allocation, and precoding of the cellular system. By simulations, we demonstrated that the proposed method provides considerable gains in the ergodic sum-spectral efficiencies compared to the existing precoding methods for numerous system configurations. 

One interesting direction would be to extend the proposed algorithm for mmWave systems by considering hybrid precoding methods in \cite{Amehd_Heath_2015} and joint transmission methods for C-RANs in \cite{Han_Lee}. Another interesting research direction would be to extend the proposed algorithm for non-orthogonal multiple access downlink precoding methods in \cite{Shin2017}.

% One interesting direction would be to extend the proposed optimization framework to multi-cell MIMO cooperative networks when imperfect CSIT is available. This direction could be useful in understanding the limits of linear preprocessing for multi-cell MIMO cooperative networks when perfect CSIT is unavailable. Another interesting research direction would be to extend the proposed algorithm for mmWave systems considering hybrid precoding methods in \cite{Amehd_Heath_2015,Dai_Clerckx}.

%%%%%%%%%%%%%%%%%%%%%%%%%%%%%%%%%%%%%%%%%%%%%%%%%%%%%%%%%%%%%%%%%%%%%

\end{document}